\documentclass[conference]{IEEEtran}

\usepackage{fancyhdr}

\usepackage{color}
\usepackage[cmex10]{amsmath} % Do not use bitmap equations (even in footnotes etc.) 
\usepackage{amssymb}
\usepackage{bbm} 
\usepackage{graphicx} 
\usepackage[caption=false,font=footnotesize]{subfig} % Use IEEE captions 
\usepackage[rightcaption]{sidecap} 
\usepackage{floatrow}
\usepackage{tabularx}
\usepackage{slashbox}
\usepackage{engord} 
\usepackage{datetime} 
\usepackage{enumitem}
\usepackage{xifthen} 
\usepackage{multicol} 
\usepackage{amsthm} 

\usepackage[linesnumbered,ruled,nofillcomment]{algorithm2e}
\SetKwBlock{KwFunction}{function}{}
\SetKwBlock{KwOn}{on}{}
\SetKwBlock{KwPeriodically}{periodically}{}
\SetKwBlock{KwState}{state}{}
\SetKwBlock{KwAtomic}{atomically}{}
\SetKwBlock{KwConcurrently}{concurrently}{}
\SetCommentSty{textup} 
\SetKwComment{KwIttayComment}{(}{)} 
\SetKw{KwAnd}{and} 
\SetKw{KwOr}{or}
\SetKw{KwSetTimer}{setTimer}
\SetKwBlock{KwExpTimer}{on timer expiration}{}
\SetKwFor{ForAll}{forall}{}

\let\oldnl\nl% Store \nl in \oldnl
\newcommand{\nonl}{\renewcommand{\nl}{\let\nl\oldnl}}% Remove line number for one line

\newcommand{\KwTitle}[1]{
    \nonl \centerline{\textbf{#1}} 
}

\ifx\pdftexversion\undefined
\usepackage[colorlinks,linkcolor=black,filecolor=black,citecolor=black,urlcolor=black,pdfstartview=FitH]{hyperref}
\else
\usepackage[colorlinks,linkcolor=blue,filecolor=blue,citecolor=blue,urlcolor=blue,pdfstartview=FitH]{hyperref}
\fi

\newcommand{\eqDef}{\stackrel{\Delta}=}

\newtheorem{definition}{Definition}
\newtheorem{lemma}{Lemma}

\newcommand{\negspace}{\vspace{-0.5\baselineskip}}

\newcommand{\xij}[2]{\ensuremath{ x_{#1,#2} }}
\newcommand{\xijbar}[2]{\ensuremath{ \bar{x}_{#1,#2} }}
 
\def\btc{%
  \leavevmode
  \vtop{\offinterlineskip %\bfseries
    \setbox0=\hbox{B}%
    \setbox2=\hbox to\wd0{\hfil\hskip-.03em
    \vrule height .3ex width .15ex\hskip .08em
    \vrule height .3ex width .15ex\hfil}
    \vbox{\copy2\box0}\box2}}

\title{ 
The Miner's Dilemma 
} 

\author{Ittay Eyal \\ Cornell University} 
           
\begin{document} 

\maketitle

% \pagestyle{fancy}
% \fancyhf{}
% \renewcommand{\headrulewidth}{0pt}
% \lfoot{\ittayComment{Ittay Eyal (ittay.eyal@cornell.edu)} \hfill \thepage} 

\begin{abstract} 
An open distributed system can be secured by requiring participants to present proof of work and rewarding them for participation. 
The Bitcoin digital currency introduced this mechanism, which is adopted by almost all contemporary digital currencies and related services. 

A natural process leads participants of such systems to form pools, where members aggregate their power and share the rewards. Experience with Bitcoin shows that the largest pools are often open, allowing anyone to join. 
It has long been known that a member can sabotage an open pool by seemingly joining it but never sharing its proofs of work. 
The pool shares its revenue with the attacker, and so each of its participants earns less. 

We define and analyze a game where pools use some of their participants to infiltrate other pools and perform such an attack. 
With any number of pools, no-pool-attacks is not a Nash equilibrium. 
With two pools, or any number of identical pools, there exists an equilibrium that constitutes a tragedy of the commons where the pools attack one another and all earn less than they would have if none had attacked. 

For two pools, the decision whether or not to attack is the miner's dilemma, an instance of the iterative prisoner's dilemma. The game is played daily by the active Bitcoin pools, which apparently choose not to attack. If this balance breaks, the revenue of open pools might diminish, making them unattractive to participants. 
\end{abstract} 

%%%%%%%%%%%%%%%%%%%%%%%%%%%%%%%%%%%%%%%%%%%%%%%%%%%%%%%%%%%%%%%%%%%%%%%%%%%%%%%
%%%%%%%%%%%%%%%%%%%%%%%%%%%%%%%%%%%%%%%%%%%%%%%%%%%%%%%%%%%%%%%%%%%%%%%%%%%%%%%
%%%%%%%%%%%%%%%%%%%%%%%%%%%%%%%%%%%%%%%%%%%%%%%%%%%%%%%%%%%%%%%%%%%%%%%%%%%%%%%

    \section{Introduction} 

Bitcoin~\cite{nakamoto2008bitcoin} is a digital currency that is gaining acceptance~\cite{soper2014paypal} and recognition~\cite{chowdhry2014google}, with an estimated market capitalization of over~4.5 billion US dollars, as of November~2014~\cite{blockchain2014marketCap}. 
Bitcoin's security stems from a robust incentive system. Participants are required to provide expensive proofs of work, and they are rewarded according to their efforts. This architecture has proved both stable and scalable, and it is used by most contemporary digital currencies and related services, e.g.~\cite{litecoin2013site, dogecoin2013site, miller2014permacoin, namecoin2013site}. Our results apply to all such incentive systems, but we use Bitcoin terminology and examples since it serves as an active and prominent prototype. 

Bitcoin implements its incentive systems with a data structure called the \emph{blockchain}. 
The blockchain is a serialization of all money transactions in the system. It is a single global ledger maintained by an open distributed system. 
Since anyone can join the open system and participate in maintaining the blockchain, Bitcoin uses a \emph{proof of work} mechanism to deter attacks: participation requires exerting significant compute resources. 
A participant that proves she has exerted enough resources with a proof of work is allowed to take a step in the protocol by generating a block. 
Participants are compensated for their efforts with newly minted Bitcoins.  The process of creating a block is called \emph{mining}, and the participants~--- \emph{miners}. 

In order to win the reward, many miners try to generate blocks. 
The system automatically adjusts the \emph{difficulty} of block generation,
such that one block is added every~10 minutes to the blockchain. 
This means that each miner seldom generates a block. 
Although its revenue may be positive in expectation, a miner may have to wait for an extended period to create a block and earn the actual Bitcoins. 
Therefore, miners form \emph{mining pools}, where all members mine concurrently and they share their revenue whenever one of them creates a block. 

Pools are typically implemented as a \emph{pool manager} and a cohort of miners. The pool manager joins the Bitcoin system as a single miner. 
Instead of generating proof of work, it outsources the work to the miners. 
In order to evaluate the miners' efforts, the pool manager accepts partial proof of work and estimates each miner's \emph{power} according to the rate with which it submits such partial proof of work. 
When a miner generates a full proof of work, it sends it to the pool manager which publishes this proof of work to the Bitcoin system. 
The pool manager thus receives the full revenue of the block and distributes it fairly according to its members power. 
Many of the pools are open~--- they allow any miner to join them using a public Internet interface. 

Such open pools are susceptible to the classical \emph{block withholding attack}~\cite{rosenfeld2011analysis}, where a miner sends only partial proof of work to the pool manager and discards full proof of work. 
Due to the partial proof of work it sends to the pool, the miner is considered a regular pool member and the pool can estimate its power. 
Therefore, the attacker shares the revenue obtained by the other pool members, but does not contribute. It reduces the revenue of the other members, but also its own. 
We provide necessary background on the Bitcoin protocol, pools and the classical block withholding attack in Section~\ref{sec:prelim}, and specify our model in Section~\ref{sec:model}. 

In this work we analyze block withholding attacks among pools. 
A pool that employs the \emph{pool block withholding attack} registers with the victim pool as a regular miner. 
It receives tasks from the victim pool and transfers them to some of its own miners. 
We call these \emph{infiltrating} miners, and the mining power spent by a pool the \emph{infiltration rate}. 
When the attacking pool's infiltrating miners deliver partial proofs of work, the attacker transfers them to the victim pool, letting the attacked pool estimate their power. 
When the infiltrating miners deliver a full proof of work, the attacking pool discards it. 

This attack affects the revenues of the pools in several ways. 
The victim pool's effective mining rate is unchanged, but its total revenue is divided among more miners. 
The attacker's mining power is reduced, since some of its miners are used for block withholding, but it earns additional revenue through its infiltration of the other pool. 
And finally, the total effective mining power in the system is reduced, causing the Bitcoin protocol to reduce the difficulty. 

Taking all these factors into account, we observe that a pool might be able to increase its revenue by attacking other pools. Each pool therefore makes a choice of whether to attack each of the other pools in the system, and with what infiltration rate. This gives rise to the \emph{pool game}. We specify this game and provide initial analysis in Section~\ref{sec:poolGame}. 

% The total effective mining power of the victim pool is unchanged, since the attack does not affect its mining power. However, the victim pool distributes this revenue among its loyal miners, as well as the infiltrating miners. Therefore the revenue of each of the victim's loyal miners is reduced. The revenue of the attacking pool does change, since some of its miners are used for infiltration rather than mining. Its total revenue is the revenue from direct mining by its loyal miners, plus the revenue received from the victim pool due to the infiltrating miners' proofs of work. 

% We begin our analysis with a simple scenario, where there are two pools only one of which can attack the other, and learn that the attacker can increase its revenue by attacking. Each pool therefore makes a choice of whether to attack each of the other pools in the system, and with what infiltration rate. This gives rise to the \emph{pool game}. We specify this game in Section~\ref{sec:poolGame}. 

In Section~\ref{sec:twoPoolsOneAttacker} we analyze the scenario of exactly two pools where only one can attack the other. 
Here, the attacker can always increase its revenue by attacking. 
We conclude that in the general case, with any number of pools, no-pool-attacks is not a Nash equilibrium. 

Next, Section~\ref{sec:twoPools} deals with the case of two pools, where each can attack the other. Here analysis becomes more complicated in two ways. 
First, the revenue of each pool affects the revenue of the other through the infiltrating miners. We prove that for a static choice of infiltration rates the pool revenues converge. 
Second, once one pool changes its infiltration rate of the other, the latter may prefer to change its infiltration rate of the former. 
Therefore the game itself takes steps to converge. 
We show analytically that the game has a single Nash Equilibrium and numerically study the equilibrium points for different pool sizes. 
For pools smaller than $50\%$, at the equilibrium point both pools earn less than they would have in the non-equilibrium no-one-attacks strategy. 

Since pools can decide to start or stop attacking at any point, this can be modeled as the \emph{miner's dilemma}~--- an instance of the iterative prisoner's dilemma. Attacking is the dominant strategy in each iteration, but if the pools can agree not to attack, both benefit in the long run. 

Finally we address the case of an arbitrary number of identical pools in Section~\ref{sec:pPools}. 
There exists a symmetric equilibrium point in which each pool attacks each other pool. 
As in the minority two-pools scenario, here too at equilibrium all pools earn less than with the no-pool-attacks strategy. 

Our results imply that block withholding by pools leads to an unfavorable equilibrium. Nevertheless, due to the anonymity of miners, a single pool might be tempted to attack, leading the other pools to attack as well. The implications might be devastating for open pools: If their revenues are reduced, miners will prefer to form closed pools that cannot be attacked in this manner. 
Though this may be conceived as bad news for public mining pools, on the whole it may be good news to the Bitcoin system, which prefers small pools. 
We discuss this and other issues pertaining to practice in Section~\ref{sec:discussion}. 

In summary, our contributions are the following: 
\begin{enumerate} 

\item Definition of the pool game where pools in a proof-of-work secured system attack one another with a pool block withholding attack. 

\item In the general case, no-pool-attacks is not an equilibrium. 

\item With two minority pools, the only Nash Equilibrium is when the pools attack one another, and both earn less than if none had attacked. 

Miners therefore face the miner's dilemma, an instance of the iterative prisoner's dilemma, repeatedly choosing between attack and no-attack. 

\item With multiple pools of equal size there is a symmetric Nash equilibrium, where all pools earn less than if none had attacked. 

\item For Bitcoin, inefficient equilibria for open pools may serve the system by reducing their attraction and pushing miners towards smaller closed pools. 

\end{enumerate} 

% Prior work has demonstrated that the analysis of Bitcoin's security is tightly tied to its incentive system~\cite{babaioff2012baloons, eyal2013majority}. 
The classical block withholding attack is old as pools themselves, but its use by pools has not been suggested until recently. We overview related attacks and prior work in Section~\ref{sec:related}, and end with concluding remarks in Section~\ref{sec:conclusion}. 

%%%%%%%%%%%%%%%%%%%%%%%%%%%%%%%%%%%%%%%%%%%%%%%%%%%%%%%%%%%%%%%%%%%%%%%%%%%%%%% 
%%%%%%%%%%%%%%%%%%%%%%%%%%%%%%%%%%%%%%%%%%%%%%%%%%%%%%%%%%%%%%%%%%%%%%%%%%%%%%% 
%%%%%%%%%%%%%%%%%%%%%%%%%%%%%%%%%%%%%%%%%%%%%%%%%%%%%%%%%%%%%%%%%%%%%%%%%%%%%%% 

    \section{Preliminaries --- Bitcoin and Pooled Mining} \label{sec:prelim}

Bitcoin is a distributed, decentralized digital currency~\cite{bitcoin2013protocol,bitcoin2013rules,nakamoto2008bitcoin,bitcoin2013source}. 
Clients use the system by issuing transactions, and the system's only task is to serialize transactions in a single ledger and reject transactions that cannot be serialized due to conflicts with previous transactions. 
Bitcoin transactions are protected with cryptographic techniques that ensure that only the rightful owner of a Bitcoin can transfer it. 

The transaction ledger is stored in a data structure caller the \emph{blockchain}. 
The blockchain is maintained by a network of \emph{miners}, which are compensated for their effort in Bitcoins. The miners are in charge of recording the transactions in the blockchain. 

        \subsection{Revenue for Proof Of Work} 

The blockchain records the transactions in units of blocks. 
Each block includes a unique ID, and the ID of the preceding block. 
The first block, dubbed \emph{the genesis block}, is defined as part of the protocol. 
A valid block contains the hash of the previous block, the hash of the transactions in the current block, and a Bitcoin address which is to be credited with a reward for generating the block. 

Any miner may add a valid block to the chain by (probabilistically) proving that it has spent a certain amount of work and publishing the block with the proof over an overlay network to all other miners. 
When a miner creates a block, it is compensated for its efforts with Bitcoins. 
This compensation includes a per-transaction fee paid by the users whose transactions are included, and an amount of minted Bitcoins that are thus introduced into the system. 
The rate at which the new Bitcoins are generated with each block is designed to slowly decrease towards zero, and will reach zero when~21 million Bitcoins are created. Then, the miners' revenue will be only from transaction fees. 

The work which a miner is required to do is to repeatedly calculate a a hash function~--- specifically the SHA-256 of the SHA-256 of a block header. To indicate that he has performed this work, the miner provides a probabilistic proof as follows. The generated block has a nonce field, which can contain any value. The miner places different values in this field and calculates the hash for each value. If the result of the hash is smaller than a target value, the nonce is considered a solution, and the block is valid. 

The number of attempts to find a single hash is therefore random with a geometric distribution, as each attempt is a Bernoulli trial with a success probability determined by the target value. At the existing huge hashing rates and target values, the time to find a single hash can be approximated by an exponential distribution. The average time for a miner to find a solution is therefore proportional to its hashing rate or \emph{mining power}. 

To maintain a constant rate of Bitcoin generation, and as part of its defense against denial of service and other attacks, the system normalizes the rate of block generation. To achieve this, the protocol deterministically defines the target value for each block according to the time required to generate recent blocks. The target, or \emph{difficulty}, is updated once every~2016 blocks such that the average time for each block to be found is~10 minutes. 

Note that the exponential distribution is memoryless. If all miners mine for block number $b$, once the block is found at time $t$, all miners switch to mine for the subsequent block $b+1$ at $t$ without changing their probability distribution of finding a block after~$t$. Therefore, the probability that a miner $i$ with mining power $m_i$ finds the next block is its ratio out of the \emph{total mining power} $m$ in the system. 

        \subsection*{Forks} 
            
Block propagation in the overlay network takes seconds, whereas the average mining interval is~10 minutes. It is therefore possible for two miners to generate competing blocks, both of which list the same block as their predecessor. The system has mechanism to solve such situations, causing one of the blocks to be discarded. However, such bifurcations are rare and occur on average once every~60 blocks~\cite{decker2013propagation}, and we ignore them for the sake of simplicity. Since the choice of the discarded block on bifurcation is random, one may incorporate this event into the probability of finding a block, and consider instead the probability of finding a block that is not discarded. 

        \subsection{Pools} 

As the value of Bitcoin rose, Bitcoin mining has become a rapidly advancing industry. 
Technological advancements lead to ever more efficient hashing ASICs~\cite{taylor2013bespoke}, and mining datacenters are built around the world~\cite{popper2013mines}. 
Mining is only profitable using dedicated cutting edge mining rigs, otherwise the energy costs exceed the expected revenue. 

Although expected revenue from mining is proportional to the power of the mining rigs used, a single home miner using a small rig is unlikely to mine a block for years~\cite{swanson2013calculator}. 
Consequently, miners often organize themselves into mining \emph{pools}. 
Logically, a pool is a group of miners that share their revenues when one of them successfully mines a block. For each block found, the revenue is distributed among the pool members in proportion to their mining power\footnote{This is a simplification that is sufficient for our analysis. The intricacies of reward systems are explained in~\cite{rosenfeld2011analysis}.}. 
The expected revenue of a pool member is therefore the same as its revenue had it mined \emph{solo}. 
However, due to the large power of the pool, it finds blocks at a much higher rate, and so the frequency of revenue collection is higher, allowing for a stable daily or weekly income. 

In practice, most pools are controlled by a pool manager.\footnote{A notable exception is P2Pool~\cite{forrsetv2011p2pool}, which we discuss in Section~\ref{sec:discussion}.} 
Miners register with the pool manager and mine on its behalf: The pool manager generates tasks and the miners search for solutions based on these tasks that can serve as proof of work. Once they find a solution, they send it to the pool manager. 
The pool manager behaves as a single miner in the Bitcoin system. Once it obtains a legitimate block from one of its miners, it publishes it. 
The block transfers the revenue to the control of the pool manager. 
The pool manager then distributes the revenue among the miners according to their mining power. 
The architecture is illustrated in Figure~\ref{fig:threeHonestPools}

\begin{figure}[!t]
\centering
\includegraphics[width=\linewidth]{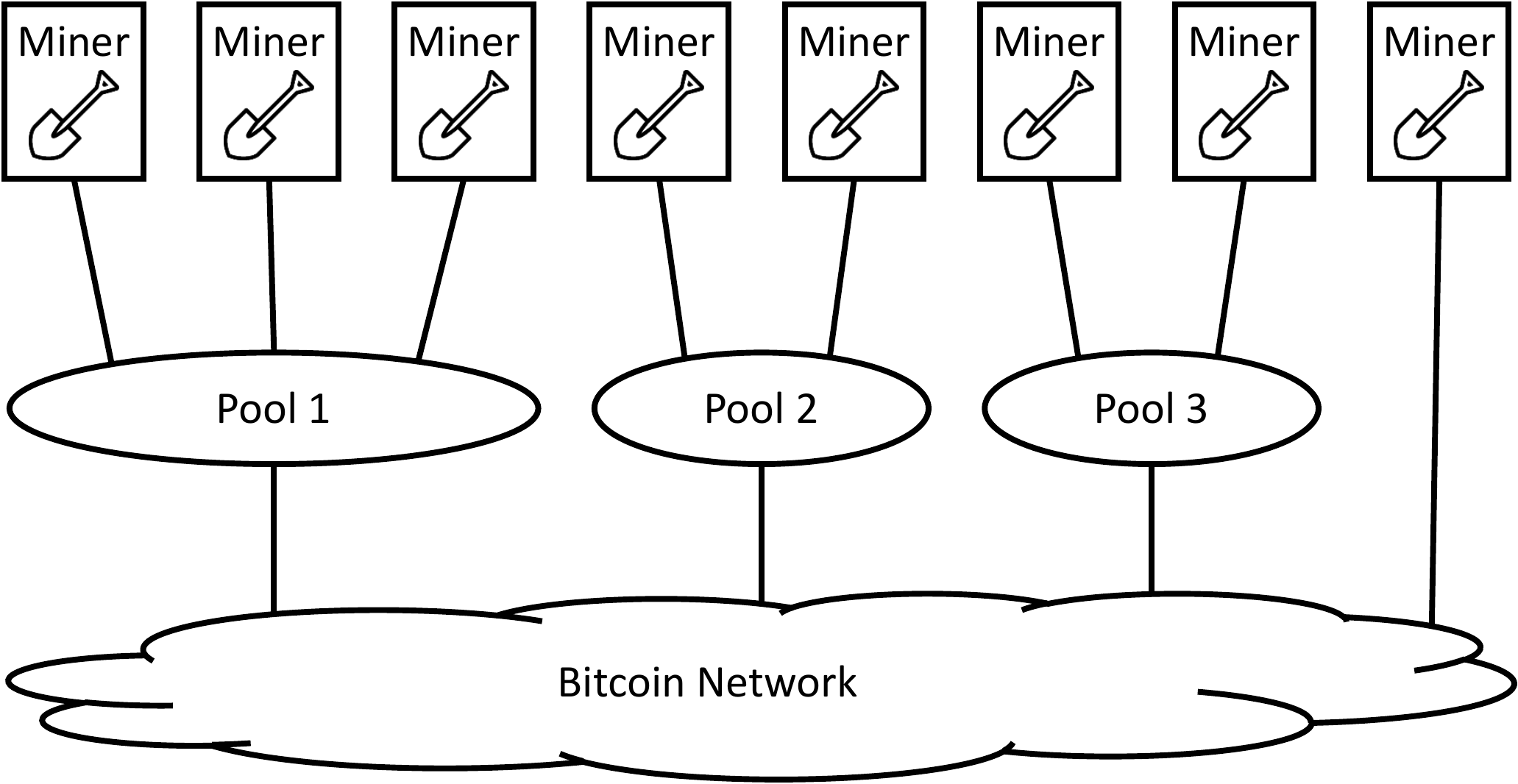}
\caption[.]{\protect
A system with 8 miners and~3 honest pools. Pool~1 has~3 registered miners, pools~2 and~3 have~2 registered miners each, and one miner mines solo. 
} 
\label{fig:threeHonestPools}
\end{figure}

In order to estimate the mining power of a miner, the pool manager sets a partial target for each member, much larger (i.e., easier) than the target of the Bitcoin system. 
Each miner is required to send the pool manager blocks that are correct according to the partial target. The partial target is chosen to be large, such that partial solutions arrive frequently enough for the manager to accurately estimate the power of the miner, but small (hard) to reduce management overhead. 
Pools often charge a small percentage of the revenue as fee. We discuss in Section~\ref{sec:discussion} the implications of such fees to our analysis. 

Many pools are open and accept any interested miner. 
Pool interface typically includes a web interface for registration and a miner interface for the mining software. 
In order to mine for a pool, a miner registers with the web interface, supplies a Bitcoin address to receive its future shares of the revenue, and receives from the pool credentials for mining. 
Then he feeds his credentials and the pool's address to its mining rig, which starts mining. The mining rig obtains its tasks from the pool and sends partial and full proof of work with the STRATUM protocol~\cite{btcWiki2014stratum}. 
As it finds blocks, the pool manager credits the miner's account according to its share of the work, and transfers these funds either on request or automatically to the aforementioned Bitcoin address. 

        \subsection*{Too Big Pools} 

Arguably in realistic scenarios of the Bitcoin system no pool controls a majority of the mining power. The reason is that the manager of a pool of this size can single-handedly take control of the Bitcoin system by generating the longest chain and ignoring blocks generated by other miners. 

If the system reaches this situation it is severely unstable~\cite{andresen2014centralized} (and~\cite{eyal2013majority} warns that the system is unstable with even smaller pools). For one day in June~2014 a single pool called GHash.IO produced over $50\%$ of the blocks in the Bitcoin main chain. The Bitcoin community backlashed at the pool (which did nothing worse than being extremely successful). GHash.IO reduced its relative mining power and publicly committed to stay away from the $50\%$ limit. 

        \subsection{Block Withholding} \label{sec:classicalBWA}

Classical Block Withholding~\cite{rosenfeld2011analysis} is an attack performed by a pool member against the other pool members. The attacking miner registers with the pool and apparently starts mining honestly~--- it regularly sends the pool partial proof of work. 
However, the attacking miner sends only partial proof of work. If it finds a full solution that constitutes a full proof of work it discards the solution, reducing the pool's total revenue. This attack is illustrated in Figure~\ref{fig:regularBWA}. 

\begin{figure}[!t]
\centering
\includegraphics[width=\linewidth]{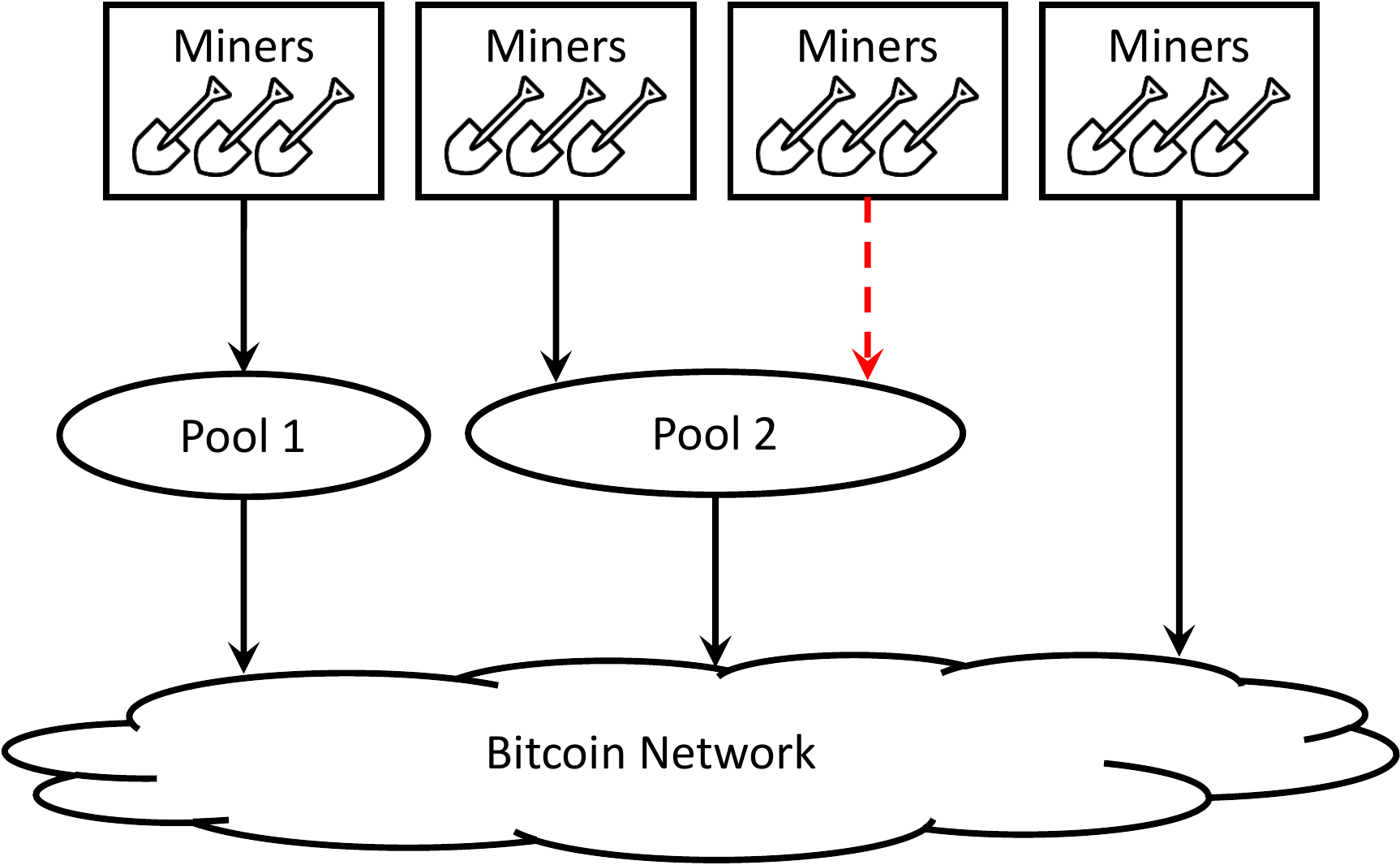}
\caption[.]{\protect
Classical Block Withholding attack. A group of miners attack Pool~2 with a block withholding attack, denoted by a dashed red arrow. 
} 
\label{fig:regularBWA} 
\end{figure} 

The attacker does not change the pool's effective mining power, and does not affect directly the revenue of other pools. However, the attacked pool shares its revenue with the attacker. Therefore each miner earns less, as the same revenue is distributed among more miners. 

Recall that the proof of work is only valid to a specific block, as it is the nonce with which the block's hash is smaller than the target. The attacking miner cannot use it. 
Moreover, this attack reduces the attacker's revenue compared to solo mining or honest pool participation: 
It suffers from the reduced revenue like the other pool participants, and its revenue is less than its share of the total mining power in the system. 
This attack can therefore only be used for sabotage, at a cost to the attacker. 

            \subsubsection*{Detection} 

Although a pool can detect that it is under a block withholding attack with good accuracy, it might not be able to detect which of its registered miners are the perpetrators. The reason is that, by design, the partial proof of work difficulty is much easier than the full proof of work difficulty. 
A pool can estimate its expected mining power and its actual mining power by the rates of partial proofs of work and full proofs of work, respectively, supplied by its miners. A difference above a set confidence interval indicates an attack. 

To detect whether a single miner is attacking it, the pool must use a similar technique, comparing the estimated mining power of the attacker based on its partial proof of work with the fact it never supplies a full proof of work. If the attacker has a small mining power, it will send frequent partial proofs of work, but the pool will only expect to see a full proof of work at very low frequency. Therefore, it cannot obtain statistically significant results that would indicate an attack. 

An attacker can therefore use multiple small miners, rather than a single large one, and replace them frequently. For example, miners whose expected full proof of work frequency is yearly will see a non-negligible daily revenue ($\btc 25 / 12 / 31 \approx \btc 0.07$). Replacing them monthly will not allow a pool to confidently tag them as attackers without tagging legitimate miners as well. 

% We normalize mining power such that the total mining power in the system is one and we normalize revenue such that the total revenue in the system is one. Consider a pool attacked by a miner with a block withholding attack. Denote the total mining power registered to the pool by~$\alpha$, of which~$\beta$ belongs to the attacker. 
% 
% The revenue that the pool expects to see is~$\alpha$, according to its registered mining power, but the actual mining power working in its behalf is only $\alpha - \beta$ as the attacker withholds blocks that are full proof of work. 
% \ittayComment{Inaccurate: need to remove the full blocks from the expected revenue of the pool, as it only receives partial proofs from the miners; this is negligible, though. } 

%%%%%%%%%%%%%%%%%%%%%%%%%%%%%%%%%%%%%%%%%%%%%%%%%%%%%%%%%%%%%%%%%%%%%%%%%%%%%%% 
%%%%%%%%%%%%%%%%%%%%%%%%%%%%%%%%%%%%%%%%%%%%%%%%%%%%%%%%%%%%%%%%%%%%%%%%%%%%%%% 
%%%%%%%%%%%%%%%%%%%%%%%%%%%%%%%%%%%%%%%%%%%%%%%%%%%%%%%%%%%%%%%%%%%%%%%%%%%%%%% 

    \section{Model and Standard Operation} \label{sec:model} 

Bitcoin is the first widely used system that rewards participants for proof of work at a dynamically normalized rate. 
Its success demonstrates the strength of this architecture, and others were quick to follow. Most are digital currencies with a various proof of work algorithms; some have other uses, for example NameCoin~\cite{namecoin2013site}, which is a DNS replacement with no central authority. 
This model, and therefore our results, apply to all such systems. 
Nevertheless, since Bitcoin is both a prototype and a working example, we use the Bitcoin terminology. 

We specify the basic model in which participants operate in Section~\ref{sec:baseModel}, proceed to describe how honest miners operate in this environment in Sections~\ref{sec:solo} and~\ref{sec:pools}, and how the classical block withholding attack is implemented with our model in Section~\ref{sec:bwMiner}. 

\newcommand{\cmdNewTask}[1]{ \ensuremath{ \texttt{newTask}(#1) }}
\newcommand{\cmdPublish}[2]{ \ensuremath{ \texttt{publish}(#1, #2) }}
\newcommand{\cmdWork}[1]{ \ensuremath{ \texttt{work}(#1) }} 
\newcommand{\send}[2]{ \ensuremath{ \texttt{send}(#1, #2) }} 
\newcommand{\recv}[1]{ \ensuremath{ \texttt{recv}(#1) }} 
\newcommand{\cmdTaskFromPool}[2]{ \ensuremath{ \texttt{taskFromPool}(#1, #2) }} 
\newcommand{\cmdPoWToPool}[3]{ \ensuremath{ \texttt{PoWToPool}(#1, #2, #3) }}
\newcommand{\cmdPoWFromMiner}[1]{ \ensuremath{ \texttt{getPoW}(#1) }}
\newcommand{\cmdSendTask}[2]{ \ensuremath{ \texttt{sendTask}(#1, #2) }}
\newcommand{\cmdPay}[2]{ \ensuremath{ \texttt{pay}(#1, #2) }}
\newcommand{\cmdCollect}[2]{ \ensuremath{ \texttt{collect}(#1, #2) }}

\newcommand{\fPoW}{ \ensuremath{ \textit{fPoW} }} 
\newcommand{\pPoW}{ \ensuremath{ \textit{pPoW} }} 
\newcommand{\tasks}{ \ensuremath{ \textit{tasks} }} 
\newcommand{\pPoWs}{ \ensuremath{ \textit{pPoWs} }} 

        \subsection{Model} \label{sec:baseModel}

The system is comprised of the Bitcoin network and nodes, and progresses in steps. 
Each node has a unique ID \textit{id} and can generate \emph{tasks} by calling the $\cmdNewTask{\textit{id}}$ command. The task is associated with the given $\textit{id}$. 

A node can work on a task for the duration of a step using the $\cmdWork{\textit{task}}$ command. 
A node that works on tasks with \cmdWork{}\ is called a miner. 
The command returns a set of partial proof of work and a set of full proofs of work. 
The number of proofs in each set has a Poisson distribution, partial proofs with a large mean and full proofs with a small mean. 
All nodes have identical power, and hence identical probabilities to generate proofs of work. 

The Bitcoin network pays for full proofs of work. To acquire this payoff an entity publishes a task $\textit{task}$ and its corresponding proof of work $\textit{PoW}$ to the network by calling the $\cmdPublish{\textit{task}}{\textit{PoW}}$ command, which returns the amount earned. The payoff goes to the ID associated with \textit{task}. 
The Bitcoin protocol normalizes revenue such that the average total revenue distributed in each step is a constant throughout the execution of the system. 
A node can transact $b$ Bitcoins to another node with ID $w$ with the $\cmdPay{w}{b}$ command. 

Apart from the \cmdWork{}\ command, all local operations, payments, message sending, propagation, and receipt are instantaneous. 

We assume that the number of miners is large enough such that mining power can be split arbitrarily without resolution constraints. Denote the number of pools with $p$, the total number of mining power in the system with $m$ and the miners loyal to pool~$i$ ($1 \le i \le p$) with $m_i$. 
We use a quasi-static analysis where miner loyalty to a pool does not change over time. 

        \subsection{Solo Mining} \label{sec:solo}

Nodes are defined by an implementation of a doStep function~--- their routine in a single step. As a first example we provide the behavior of a solo miner. 
A solo miner is a node that generates its own tasks and publishes them to earn the payoff. The algorithm of a solo miner is given in Algorithm~\ref{alg:solo}. 

\begin{algorithm}[t] 
\caption{Solo Miner $w$} 
\label{alg:solo} 
\SetAlgoNoLine
\SetAlgoNoEnd
\DontPrintSemicolon
\SetNoFillComment

\KwFunction({$\text{doStep}$}){ 
    $\textit{task} \gets \cmdNewTask{w}$ \; 
    $(\pPoW, \fPoW) \gets \cmdWork{\textit{task}}$ \; 
    $\cmdPublish{\textit{task}}{\fPoW}$ \; 
} 

\end{algorithm}

        \subsection{Pools} \label{sec:pools}

Pools are nodes that serve as coordinators and multiple miners can register to a pool and work for it. 
The pseudocode for a miner~$w$ working for a pool~$i$ is shown in Algorithm~\ref{alg:honest}. 

The pool generates the tasks and sends a task to each miner. 
The miner receives its task and works on it for the duration of the step. 
Since full proofs of work are rarely found, the pool needs to reliably measure the miner's work in a different manner. 
To facilitate pool operation, miners send the pool not only full proof of work, but also partial proofs of work. 

The pool receives the proofs of work of all its miners, registers the partial proofs of work and publishes the full proofs. 
It calculates its overall revenue, and proceeds to distribute it among its miners. 
Each miner receives revenue proportional to its success in the current step, namely the ratio of its partial proofs of work out of all partial proofs of work the pool received. The pool pays miner $w$ its share $b$ of the revenue with $\cmdPay{w}{b}$, and notifies the miner of the payment with a separate message. 
We assume that pools do not collect fees of the revenue. Pool fees and their implications on our analysis are discussed in Section~\ref{sec:discussion}. 

% A miner working solo publishes the proofs of work it finds to earn its revenue. 
% A miner working for a pool receives its revenue from the pool. The pool publishes full proofs of work it obtained from miners, and then pays each miner according to its efforts. 

\begin{algorithm}[t] 
\caption{Honest miner $w$ at Pool $i$.} 
\label{alg:honest} 
\SetAlgoNoLine
\SetAlgoNoEnd
\DontPrintSemicolon
\SetNoFillComment

\KwTitle{Miner $w$} \;
\negspace 

\KwState{ 
    $\textit{revenue} \in \mathbb{R}$, initially 0 
}
\BlankLine 

\KwFunction({$\text{doStep}$}){ 
    $\textit{task} \gets \recv{i}$ \; 
    $(\pPoW, \fPoW) \gets \cmdWork{\textit{task}}$ \; 
    \send{i}{(\pPoW, \fPoW)} \; 
    $\textit{revenue} \gets \textit{revenue} + \recv{i}$ \; 
} 
\BlankLine \BlankLine 

\KwTitle{Pool $i$} \;
\negspace 

\KwState{
    $W \subset \textup{ID space}$ \; 
    $\tasks(\cdot): W \rightarrow \textup{Task space}$ \; 
}
\BlankLine

\KwFunction({$\text{doStep}$}){ 
    $\textit{stepRevenue} \gets 0$ \; 
    \ForEach{Registered Miner $w$} { 
        $\tasks(w) \gets \cmdNewTask{i}$ \; 
        $\send{w}{\tasks(w)}$ \; 
    }
\BlankLine
    \KwIttayComment{Wait until end of step.} 
\BlankLine
    \ForEach{$w \in W$} {
        $(\pPoW, \fPoW) \gets \recv{w}$ \; 
        $\textit{stepRevenue} \gets \textit{stepRevenue} + \cmdPublish{\textit{tasks}(w)}{\textit{fPoW}}$ \; 
        $\pPoWs(w) \gets pPoW$ \; 
    } 
    $\textit{stepPPow} \gets \sum_{\text{registered } w} \pPoWs(w)$ \; 
    \ForEach{$w \in W$} { 
        $\cmdPay{w}{\textit{stepRevenue} \times \pPoWs(w) / \textit{\textit{stepPPow}}}$ \; 
        $\send{w}{\textit{stepRevenue} \times \pPoWs(w) / \textit{\textit{stepPPow}}}$
    } 
} 

\end{algorithm}

        \subsection{Block Withholding Miner} \label{sec:bwMiner} 

A miner registered at a pool can perform the classical block withholding attack, where it operates as if it worked for the pool, only it never sends its proof of work. The pseudocode is in Algorithm~\ref{alg:bwMiner}, for a miner interacting with a standard pool as in Algorithm~\ref{alg:honest}. 
The pool registers the miner's partial proofs, but cannot distinguish between miners running the honest miner of Algorithm~\ref{alg:honest} and block withholding miners running Algorithm~\ref{alg:bwMiner}. 

The implications are that a miner that engages in block withholding does not contribute to the pool's overall mining power, but still shares the pool's revenue according to its sent partial proofs of work. 

\begin{algorithm}[t] 
\caption{Block Withholding Miner $w$ at pool $i$.} 
\label{alg:bwMiner} 
\SetAlgoNoLine
\SetAlgoNoEnd
\DontPrintSemicolon
\SetNoFillComment

\KwFunction({$\text{doStep}$}){ 
    $\textit{task} \gets \cmdTaskFromPool{i}{w}$ \; 
    $(\pPoW, \fPoW) \gets \cmdWork{\textit{task}}$ \; 
    \cmdPoWToPool{i}{\pPoW}{\emptyset} \; 
    $\cmdCollect{i}{w}$ \; 
} 
\end{algorithm} 

To reason about a pool's efficiency we define its per-miner revenue as follows. 

\begin{definition}[Revenue density] 
The \emph{revenue density} of a pool is the ratio between the average revenue a pool member earns and the average revenue it would have earned as a solo miner. 
\end{definition} 

The revenue density of a solo miner, and that of a miner working with an unattacked pool are~1. If a pool is attacked with block withholding, its revenue density decreases. 

        \subsection{Continuous Analysis} \label{sec:continuous}

Because our analysis will be of the average revenue, we will consider proofs of work, both full and partial, as continuous deterministic sizes, according to their probability. 
The $\cmdWork{}$ command will therefore return a deterministic fraction of proof of work. 

%%%%%%%%%%%%%%%%%%%%%%%%%%%%%%%%%%%%%%%%%%%%%%%%%%%%%%%%%%%%%%%%%%%%%%%%%%%%%%% 
%%%%%%%%%%%%%%%%%%%%%%%%%%%%%%%%%%%%%%%%%%%%%%%%%%%%%%%%%%%%%%%%%%%%%%%%%%%%%%% 
%%%%%%%%%%%%%%%%%%%%%%%%%%%%%%%%%%%%%%%%%%%%%%%%%%%%%%%%%%%%%%%%%%%%%%%%%%%%%%% 

    \section{The Pool Game} \label{sec:poolGame}

        \subsection{The Pool Block Withholding Attack} 
 
Just as a miner can perform block withholding on a pool $j$, a pool~$i$ can use some of its miners to infiltrate a pool $j$ and perform a block withholding attack on~$j$. Denote the number of such infiltrating miners at step~$t$ by~$\xij{i}{j}(t)$. In this case, the infiltrating miners obtain their share of pool~$j$'s revenue, and transfer it back to pool~$i$. Infiltrators from~$i$ to~$j$, as well as any miners that honestly mine for pool~$i$, are \emph{loyal} to pool~$i$. Pool~$i$ distributes its revenue from mining and from its infiltrators evenly among all its registered miners, according to their partial proofs of work. 

The pseudocode of a block withholding pool is shown in Algorithm~\ref{alg:bwa}. The pool's miners are oblivious to the change and their algorithm is identical to the one in Algorithm~\ref{alg:honest}. 

\newcommand{\prevInfRev}{\ensuremath{ \textit{prevRoundInfiltrationRevenue} }}
\newcommand{\infilt}{\ensuremath{ \textit{inf} }}

\begin{algorithm*}[t] 
\caption{Block Withholding Pool $i$.} 
\label{alg:bwa} 
\SetAlgoNoLine
\SetAlgoNoEnd
\DontPrintSemicolon
\SetNoFillComment

\KwState{
    $\prevInfRev \in \mathbb{R}$, initially 0 \; 
    $W \subset \textup{ID space}$ \; 
    $\infilt(\cdot): W \rightarrow \textup{ID space} \cup {\bot}$ \; 
    $\tasks(\cdot): W \rightarrow \textup{Task space}$ \; 
}
\BlankLine

\KwFunction({$\text{doStep}$}){ 
    \ForEach{$w \in W$} { 
        \If({\hfill(infiltrator)}){$\infilt(w) \neq \bot$} { 
            $\tasks(w) \gets \cmdTaskFromPool{\infilt(w)}{i}$ \; 
        } \Else { 
            $\tasks(w) \gets \cmdNewTask{i}$ \; 
        }
        $\cmdSendTask{w}{\tasks(w)}$ \; 
    }
\BlankLine 
    \KwIttayComment{Wait until end of round} 
\BlankLine 
    $\textit{stepRevenue} \gets \prevInfRev$ \; 
    $\textit{stepPPow} \gets 0$ \; 
    \ForEach{$w \in W$} { 
        $(\pPoW, \fPoW) \gets \cmdPoWFromMiner{w}$ \; 
        $\textit{stepPPow} \gets \textit{stepPPow} + | \pPoW |$ \; 
        \If({\hfill(infiltrator)}){$\infilt(w) \neq \bot$} { 
            $\cmdPoWToPool{\infilt(w)}{\pPoW}{\emptyset}$ \; 
        } \Else { 
            $\textit{stepRevenue} \gets \textit{stepRevenue} + \cmdPublish{tasks(w)}{\fPoW}$ \; 
        } 
    } 
    \ForEach{$w \in W$} {
        $\cmdPay{w}{\textit{stepRevenue} \cdot \pPoWs(w) / \textit{stepPPow}}$ \; 
    } 
\BlankLine 
    $\prevInfRev \gets 0$ \; 
    \ForEach{$\{ w | w \in W \wedge \infilt(w) \neq \bot \}$} { 
        $\prevInfRev \gets \prevInfRev + \cmdCollect{\infilt(w)}{i}$  \; 
    } 
} 
\end{algorithm*} 

        \subsection{Block Withholding Recycling} 

We assume that the infiltrating miners are loyal to the attacker. 
However, some of the pool's members may be disloyal infiltrators. For example, pool~1 can use a loyal miner to infiltrate pool~2, and pool~2, thinking the miner is loyal to it, might use it to attack pool~1. 

When sending disloyal miners to perform block withholding at other pools, an attacker takes a significant risk. The pool~2's perspective in the scenario described above. The disloyal miner can perform honest mining for pool~1, rather than withhold its blocks, and not return any revenue to pool~2. Moreover, it will take its share of pool~2's revenues (which thinks the miner is loyal to it) and deliver it back to pool~1. 

To avoid such a risk, a pool needs a sufficient number of verified miners~--- miners that it knows to be loyal. In Bitcoin this happens in the common case where the pool owner has miners of his own. 

        \subsection{Revenue Convergence} 

Note that pool~$j$ sends its revenue to infiltrators from pool~$i$ at the end of the step, and this revenue is calculated in pool~$i$ at the beginning of the subsequent step with the \prevInfRev\ variable. If there is a chain of pools of length $\ell$ where each pool infiltrates the next, the pool revenue will not be constant, since the revenue from infiltration takes one step to take each hop. If $\ell_{\max}$ is the longest chain in the system, the revenue stabilizes after $\ell_{\max}$ steps and remains constant. If there are loops in the infiltration graph, the system will converge to a certain revenue. 

\begin{lemma}[Revenue convergence] 
If infiltration rates are constant, the pool revenues converge. 
\end{lemma} 

\begin{proof} 
Denote the revenue density of pool~$i$ at the end of step~$t$ by~$r_i(t)$, and define the revenue density vector 
\[
\textbf{r(t)} \eqDef (r_1(t), \dots, r_p(t))^T \,\,\, . 
\] 
In every round, pool~$i$ uses its mining power of $m_1 - \sum_j \xij{1}{j}$ used for direct mining (and not attacking), and shares it among its $m_1 + \sum_j \xij{j}{1}$ members (all sums are over the range $1, \dots, p$), including malicious infiltrators. 
Denote the direct mining revenue density of each pool (ignoring normalization, which is a constant factor) with the vector 
\[ 
\textbf{m} 
\eqDef 
\left( 
    \frac{m_1 - \sum_j \xij{1}{j}}{m_1 + \sum_j \xij{j}{1}}, 
    \dots, 
    \frac{m_p - \xij{p}{j}}{m_p + \sum_j \xij{j}{p}}
\right)^T \,\,\, . 
\] 

The revenue of Pool~$i$ in step $t$ taken through infiltration from pool~$j$'s revenue in step $t-1$ is $\xij{i}{j} r_j(t-1)$. Pool~i distributes this revenue among its $m_i + \sum_k \xij{k}{i}$ members~--- loyal and infiltrators. 
Define the $p \times p$ \emph{infiltration matrix} whose $i, j$ element is 
\[
\textbf{G} 
\eqDef 
% \left[ \frac{\xij{i}{j}}{m_i + \sum_j \xij{j}{i}} \right]_{ij} \,\,\, . 
\left[ \frac{\xij{i}{j}}{m_i + \sum_k \xij{k}{i}} \right]_{ij} \,\,\, . 
% \left[ \frac{\xij{i}{j}}{m_j + \sum_k \xij{k}{j}} \right]_{ij} \,\,\, . 
\]
And the revenue vector at step $t$ is 
\begin{equation} \label{eqn:rt}
\textbf{r}(t) = \textbf{m} + \textbf{G} \textbf{r}(t-1) \,\,\, . 
\end{equation}  

Since the row sums of the infiltration matrix are smaller than one, its largest eigenvalue is smaller than~1 according to the Perron-Frobenius theorem. Therefore, the revenues at all pools converges as follows for~$t \ge 1$: 
\begin{equation} \label{eqn:rt}
\textbf{r}(t) 
= 
\left( \sum_{t'=0}^{t-1} G^t \right) \textbf{m} + G^t \textbf{r}(0) 
\xrightarrow{t \rightarrow \infty} 
(1 - \textbf{G})^{-1} \textbf{m} 
\,\,\, . 
\end{equation} 

\end{proof}

        \subsection{The Pool Game} 

In the pool game pools try to optimize their infiltration rates of other pools to maximize their revenue. The overall number of miners and the number of miners loyal to each pool remain constant throughout the game. 
 
Time progresses in rounds. Let $s$ be a constant integer large enough that revenue can be approximated as its convergence limit. 
In each round the system takes~$s$ steps and then a single pool, picked with a round-robin policy, may change its infiltration rates of all other pools. The total revenue of each step is normalized to $1/s$, so the revenue per round is one. 

            \subsubsection*{Pool Knowledge} 

The pool taking a step knows the rate of infiltrators attacking it (though not their identity) and the revenue rates of each of the other pools. 
This knowledge is required to optimize a pool's revenue, as we explain in Section~\ref{sec:general}. 

A pool can estimate the rate with which it is attacked by comparing the rates of partial and full proofs of work it receives from its miners, as explained in Section~\ref{sec:classicalBWA}. 
In order to estimate the attack rates against each of the other pools, a pool can use one of two methods. 
First, pools often publish this data to demonstrate their honesty to their miners~\cite{slush2014dashboard, ghash2014dashboard, discusfish2014dashboard}. 
Second, a pool can infiltrate each of the other pools with some nominal probing mining power and measure the revenue density directly by monitoring the probe's rewards from the pool. 

        \subsection{General Analysis} \label{sec:general} 

Recall that $m_i$ is the number of miners loyal to pool~$i$. 
and $\xij{i}{j}(t)$ is the number of miners used by pool~$i$ to infiltrate pool~$j$ at step~$t$. 

The mining rate of pool~$i$ is therefore the number of its loyal miners minus the miners it uses for infiltration. 
This absolute mining rate denoted is divided by the total mining rate in the system, namely the number of all miners that do not engage in block withholding.  
Denoted the direct mining rate at step~$t$ by 
\begin{equation} \label{eqn:RiFull} 
R_i 
\eqDef 
\frac{ 
    m_i - \sum_{j = 1}^p \xij{i}{j} 
}{ 
    m - \sum_{j = 1}^p \sum_{k = 1}^p \xij{j}{k} 
} 
\end{equation}

% \[ 
% \tilde{R}_i(t) 
% \eqDef 
% m_i - \frac{m_i}{m_i + \sum_{j = 1}^p x_j^i(t)} \sum_{j = 1}^p x_i^j(t)
% \] 
% \[
% R_i(t) 
% \eqDef 
% \frac{
%     \tilde{R}_i(t) 
% }{
%     \sum_{j=1}^p \tilde{R}_j(t) + \left( m - \sum_{j=1}^p m_j \right)
% } \,\,\, .
% \] 

The revenue density of pool $i$ at the end of step~$t$ is its revenue from direct mining together with its revenue from infiltrated pools, divided by the number of its loyal miners together with block-withholding infiltrators that attack it: 
\begin{equation} \label{eqn:riFull}
r_i(t) 
= 
\frac{
R_i(t) + \sum_{j = 1}^{p} \xij{i}{j}(t) r_j(t)
}{
m_i + \sum_{j = 1}^{p} \xij{j}{i}(t) 
}
\,\,\, .
\end{equation} 

When pool~$i$ takes a step~$t$, it knows the revenue density of all other pools $r_j(t-1)$ and its total infiltration rate $\sum_{j = 1}^{p} \xij{j}{i}(t)$. 

            \subsection*{No attack} 

If no pool engages in block withholding, 
\[
\forall i, j: \xij{i}{j} = 0 \,\,\, , 
\] 
we have at all times 
\[
\forall i: r_i(t) = 1/m \,\,\, , 
\]
that is, each miner's revenue is proportional to its power, be it in a pool or working solo. 

%%%%%%%%%%%%%%%%%%%%%%%%%%%%%%%%%%%%%%%%%%%%%%%%%%%%%%%%%%%%%%%%%%%%%%%%%%%%%%%
%%%%%%%%%%%%%%%%%%%%%%%%%%%%%%%%%%%%%%%%%%%%%%%%%%%%%%%%%%%%%%%%%%%%%%%%%%%%%%%
%%%%%%%%%%%%%%%%%%%%%%%%%%%%%%%%%%%%%%%%%%%%%%%%%%%%%%%%%%%%%%%%%%%%%%%%%%%%%%%

    \section{One Attacker} \label{sec:twoPoolsOneAttacker}

We begin our analysis with a simplified game of two pools,~1 and~2, where pool~1 can infiltrate pool~2, but pool~2 cannot infiltrates pool~1. 
%
% Recall that we denote the total number of miners by $m$ and number of miners loyal to Pool~$i$ by $m_i$. 
% The rest of the $m- m_1 - m_2$ miners mine solo. 
The $m- m_1 - m_2$ miners outside both pools mine solo (or with closed pools that do not attack and cannot be attacked). 
This scenario is illustrated in Figure~\ref{fig:twoPoolsOneAttackerIllustration}. 
The dashed red arrow indicates that $\xij{1}{2}$ of pool~1's mining power infiltrates pool~2 with a block withholding attack. 

\begin{figure}[!t]
\centering
\includegraphics[width=\linewidth]{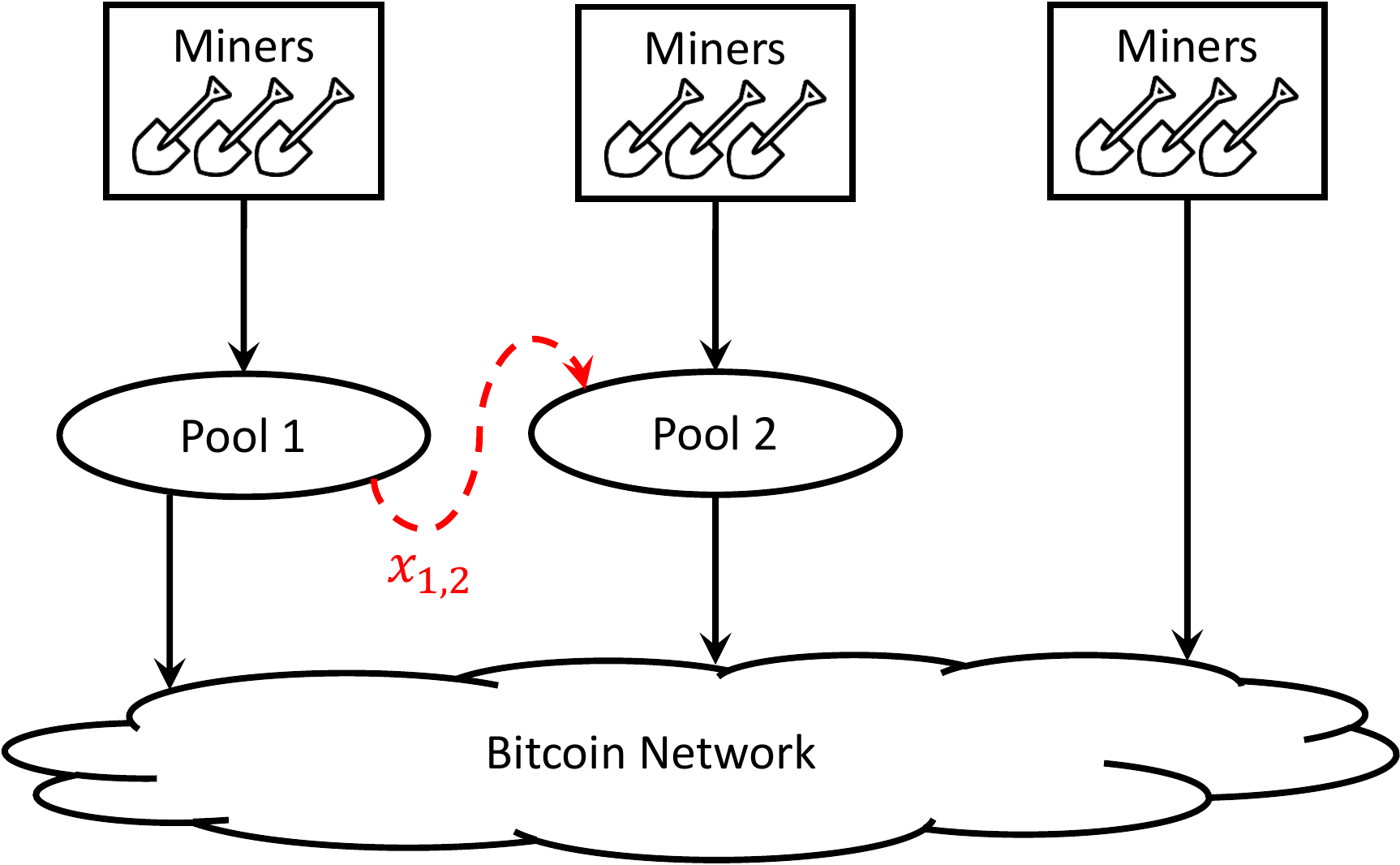}
\caption[.]{\protect 
The one-attacker scenario. Pool~1 attacks pool~2. 
} 
\label{fig:twoPoolsOneAttackerIllustration}
\end{figure}

Since Pool~2 does not engage in block withholding, all of its $m_2$ loyal miners work on its behalf. Pool~1, on the other hand does not employ $\xij{1}{2}$ of its loyal miners, and its direct mining power is only $m_1 - \xij{1}{2}$. 
The Bitcoin system normalizes these rates by the total number of miners that publish full proofs, namely all miners but $\xij{1}{2}$. 
The pools' direct revenues are therefore 
\begin{equation} \label{eqn:oneAttacker:Rs} 
\begin{aligned} 
    R_1 &= \frac{m_1 - \xij{1}{2}}{m - \xij{1}{2}} \\ 
    R_2 &= \frac{m_2}{m - \xij{1}{2}} \,\,\, . 
\end{aligned}
\end{equation} 

Pool~2 divides its revenue among its loyal miners and the miners that infiltrated it. Its revenue density is therefore  
\begin{equation} \label{eqn:oneAttacker:r2}
r_2 = 
\frac{
R_2
}{
m_2 + \xij{1}{2}
}
%= \frac{m_2}{m - x} / (m_2 + x) 
\,\,\, . 
\end{equation} 

Pool~1 divides its revenue among its registered miners. The revenue includes both its direct mining revenue and the revenue its infiltrators obtained from pool~2, which is $r_2 \cdot \xij{1}{2}$. The revenue per loyal Pool~1 miner is therefore 
\begin{equation} \label{eqn:oneAttacker:r1} 
r_1 = 
\frac{R_1 + \xij{1}{2} \cdot r_2}{m_1} 
%= \frac{\frac{m_1 - x}{m - x} + x \frac{m_2}{m - x} / (m_2 + x)}{m_1} 
\,\,\, . 
\end{equation} 

We obtain the expression for~$r_1$ in Equation~\ref{eqn:oneAttacker:r1}
by substituting $r_2$ from Equation~\ref{eqn:oneAttacker:r2} and $R_1$ and $R_2$ from equation~\ref{eqn:oneAttacker:Rs}: 
\[ 
r_1 = 
\frac{
(\xij{1}{2})^2-m_1 (m_2+\xij{1}{2})
}{
m_1 (\xij{1}{2}-1) (m_2+\xij{1}{2})
} 
\] 

        \subsection{Game Progress} 

Pool~1 controls its infiltration rate of pool~2, namely $\xij{1}{2}$, and will choose the value that maximizes the \emph{revenue density} (per-miner revenue) $r_1$ on the first round of the pool game. 

The value of $r_1$ is maximized at a single point in the feasible range $0 \le \xij{1}{2} \le m_1$. 
Since pool~2 cannot not react to pool~1's attack, this point is the stable state of the system, and we denote the value of \xij{1}{2}\ there by  
$
\xijbar{1}{2} 
\eqDef 
\arg \max_{\xij{1}{2}} r_1 
\,\,\, , 
$
and the values of the corresponding revenues of the pools with $\bar{r}_1$ and $\bar{r}_2$. 

Substituting the stable value $\xij{1}{2}$ we obtain the revenues of the two pools; all are given in Figure~\ref{fig:oneAttackerEqns}. 

\begin{figure*}[t]
\begin{equation}
\begin{aligned} 
& \xijbar{1}{2} 
=  
\frac{
m_2 - m_1 m_2 + \sqrt{-m_2^2 (-1 + m_1 + m_1 m_2)}
}{
-1 + m_1 + m_2
}
\\
& \bar{r}_1 
= 
\frac{
m_1 + (2 + m_1) m_2 + 2 \sqrt{-m_2^2 (-1 + m_1 + m_1 m_2)}
}{
m_1 (1 + m_2)^2
}
\\ 
& \bar{r}_2 
= 
\frac{
    -m_2 (-1 + m_1 + m_2)^2
}{
    \left( m_2^2 + \sqrt{-m_2^2 (-1 + m_1 + m_1 m_2)} \right) 
    \left( 1 - m_1 (1 + m_2) + \sqrt{-m_2^2 (-1 + m_1 + m_1 m_2)} \right)
}
\end{aligned}
\end{equation} 
\caption{
Stable state where only pool~1 attacks pool~2. 
} 
\label{fig:oneAttackerEqns} 
\end{figure*}

        \subsection{Numerical Analysis} 

\begin{figure*}[t]
\centering
% \Large{\tCache\ Efficiency with Realistic Workloads}

\subfloat[\xij{1}{2}]{
\includegraphics[width=0.3\linewidth]{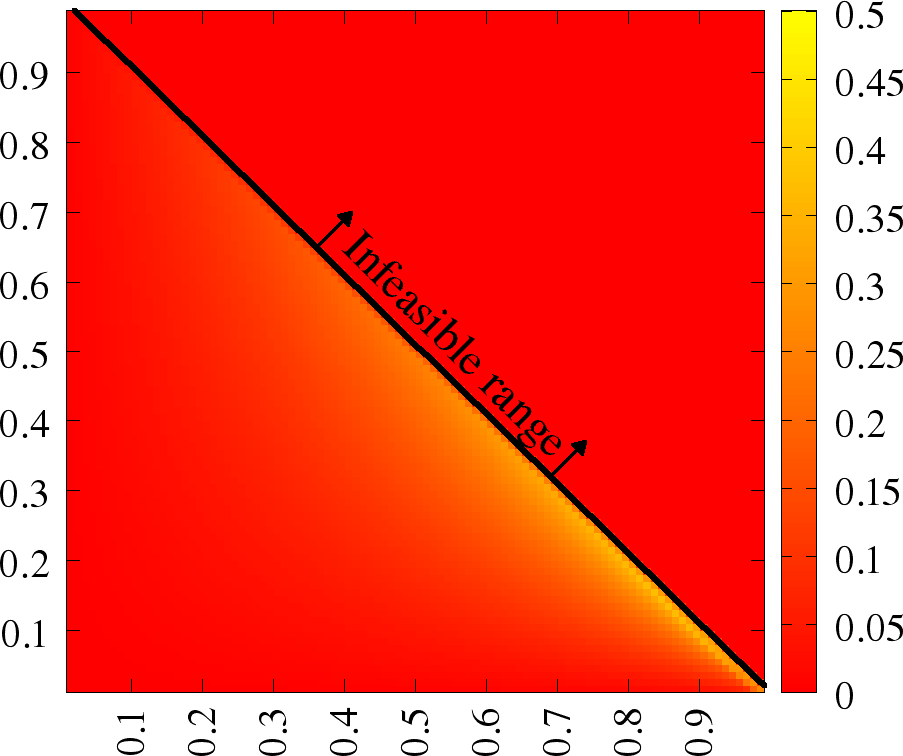}
\label{fig:twoPoolsOneAttacker:x12}
}
\hfil
\subfloat[$r_1$]{
\includegraphics[width=0.3\linewidth]{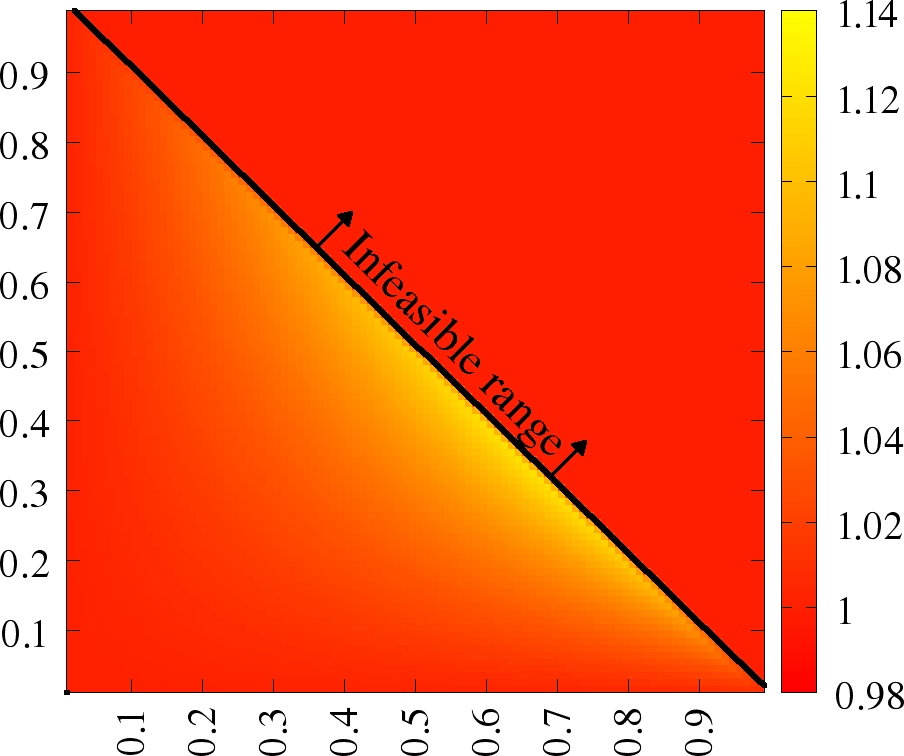}
\label{fig:twoPoolsOneAttacker:r1}
}
\hfil
\subfloat[$r_2$]{
\includegraphics[width=0.3\linewidth]{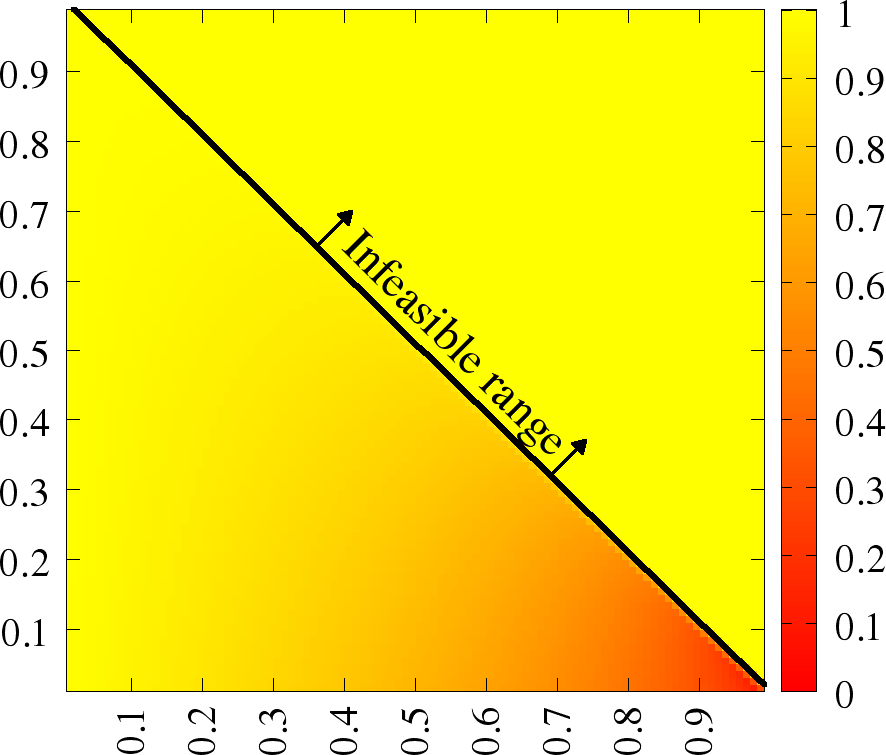}
\label{fig:twoPoolsOneAttacker:r2}
}

\caption[.]{\protect
Two pools where one infiltrates the other: Optimal infiltration rate $\xij{1}{2}$ and corresponding revenues~($r_1$ and $r_2$) as a function of pool sizes. The line in~\subref{fig:twoPoolsOneAttacker:x12} shows~$\xij{1}{2}=0$ and the lines in~\subref{fig:twoPoolsOneAttacker:r1} and~\subref{fig:twoPoolsOneAttacker:r2} show the revenue density of~1. 
} 
\label{fig:twoPoolsOneAttacker}
\end{figure*}

We analyze this game numerically by finding the $\xij{1}{2}$ that maximizes $r_1$ and substituting this value for $r_1$ and $r_2$. We vary the sizes of the pools through the entire possible range and depict the optimal $\xij{1}{2}$ and the corresponding revenues in Figure~\ref{fig:twoPoolsOneAttacker}. Each point in each graph therefore represents the equilibrium point of a game with the corresponding $m_1$ and $m_2$ sizes, where we normalize $m = 1$. The top right half of the range in all graphs is not feasible, as the sum of $m_1$ and $m_2$ is larger than~1. We use this range as a reference color, and we use a dashed line to show the bound between this value within the feasible range. 

Figure~\ref{fig:twoPoolsOneAttacker:x12} shows the optimal infiltration rate. In the entire feasible range we see that pool~1 chooses a strictly positive value for~$\xij{1}{2}$. Indeed, the revenue of pool~1 is depicted in Figure~\ref{fig:twoPoolsOneAttacker:r1} and in the entire feasible region it is strictly larger than~1, which the pool would have gotten without attacking ($\xij{1}{2} = 0$). Figure~\ref{fig:twoPoolsOneAttacker:r2} depicts the revenue of Pool~2, which is strictly smaller than~1 in the entire range. 

Note that the total system mining power is reduced when pool~1 chooses to infiltrate pool~2. Therefore, the revenue of third parties, miners not in either pool, increases from~$1/m$ to~$1/(m - \xij{1}{2})$. Pool~2 therefore pays for the increased revenue of its attacker and everyone else in the system. 

% The revenue $r_1$ is at its maximum at 
% \begin{equation} 
% \xij{1}{2}^{\max} = \frac{\sqrt{m_2^2 (m^2 - m m_1 - m_1 m_2)}-m m_2+m_1 m_2}{m-m_1-m_2} \,\,\, , 
% \end{equation}
% which respects $0 \le \xij{1}{2}^{\max} \le m_1$. The corresponding value of~$r_1$ is 
% \begin{equation} 
% r_1^{\max} = \frac{m (m_1+2 m_2)+m_1 m_2 - 2 \sqrt{m_2^2 \left(m^2-m m_1-m_1 m_2\right)}}{m_1 (m+m_2)^2} \,\,\, . 
% \end{equation} 

%         \subsection{Discussion} 
% 
% Our analysis demonstrates that a pool can always benefit from infiltrating another pool and performing a block withholding attack, sharing revenues among its loyal registrants. 
% 
% The possibility of a block withholding attack has been discussed for a long time, and in several occasions~\cite{occasion1, occasion2} pools had suspected they were being attacked. The assumed goal for such an attack was to reduce that attacked pool's revenue, possibly with the long term business goal of causing its miners to leave it in favor of the attacking pool, whose revenue remains unchanged. A recent work by courtois and Bahack~\cite{courtois2014subversive} noted that block withholding is profitable in a system of two pools (with no other miners). 
% 
% Our results demonstrate that in a system with any number of pools and independent miners that do not engage in a block withholding attack, it is beneficial for any pool to infiltrate any of its counterparts to increase its revenue. 

        \subsection{Implications to the general case} 

Consider the case of $p$ pools. For any choice of the pools sizes $m_1, \dots, m_p$, at least one pool will choose to perform block withholding: 
\begin{lemma} 
In a system with~$p$ pools, the point $\forall j, k: x_j^k = 0$ is not an equilibrium. 
\end{lemma} 

\begin{proof} 
Assume towards negation this is not the case, and $\forall j, k: x_j^k = 0$ is an equilibrium point. Now consider a setting with only pools~1 and~2, and treat the other pools as independent miners. 
This is the setting analyzed above and we have seen there that pool~1 can increase its revenue by performing a block withholding attack on pool~2. Denote pool~1's infiltration rate by~$\tilde{x}_1^2 > 0$. Now, take this values back to the setting at hand with~$p$ pools. The revenue of pool~1 is better when 
\[
x_1^2 = \tilde{x}_1^2, \forall (j, k) \neq (1, 2): x_j^k = 0 \,\,\, . 
\] 
Therefore, pool~1 can improve its revenue by attacking pool~2, and no-one-attacks is not an equilibrium point. 
\end{proof} 

~

    \section{Two Pools} \label{sec:twoPools}

\begin{figure*}[t]
\centering
% \Large{\tCache\ Efficiency with Realistic Workloads}

\subfloat[\xij{1}{2}]{
\includegraphics[width=0.4\linewidth]{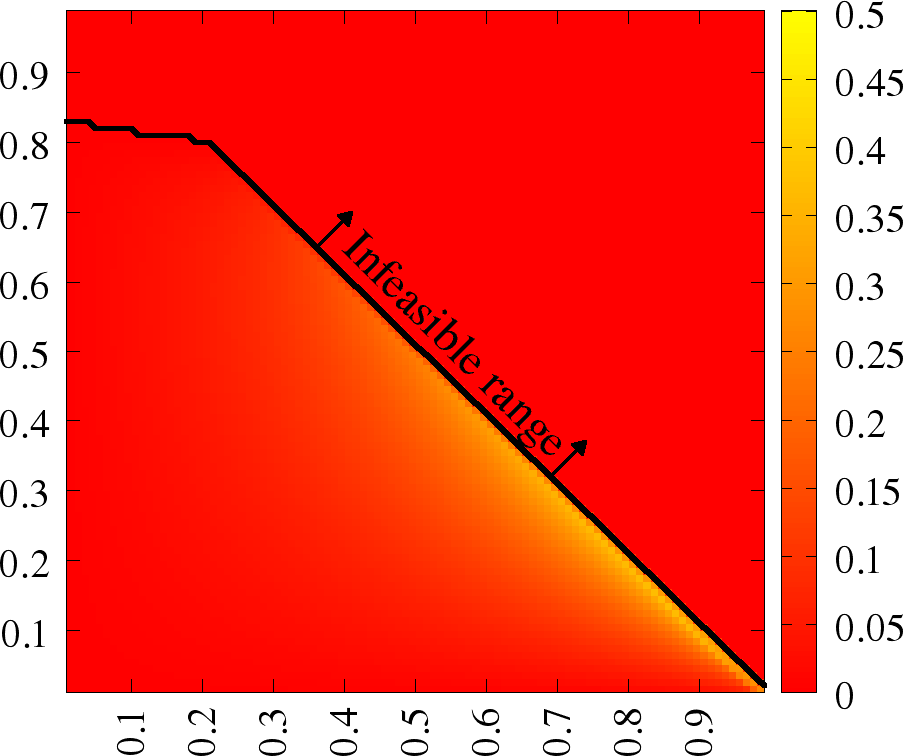}
\label{fig:twoPools:x1}
}
\hfil
\subfloat[\xij{2}{1}]{
\includegraphics[width=0.4\linewidth]{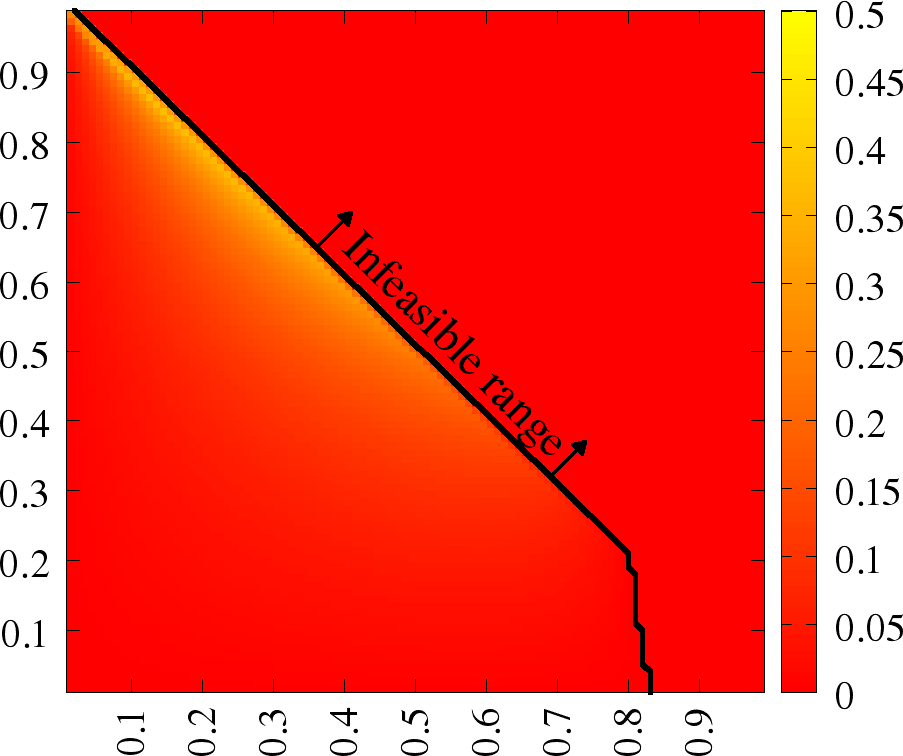}
\label{fig:twoPools:x2}
}

\subfloat[$r_1$]{
\includegraphics[width=0.4\linewidth]{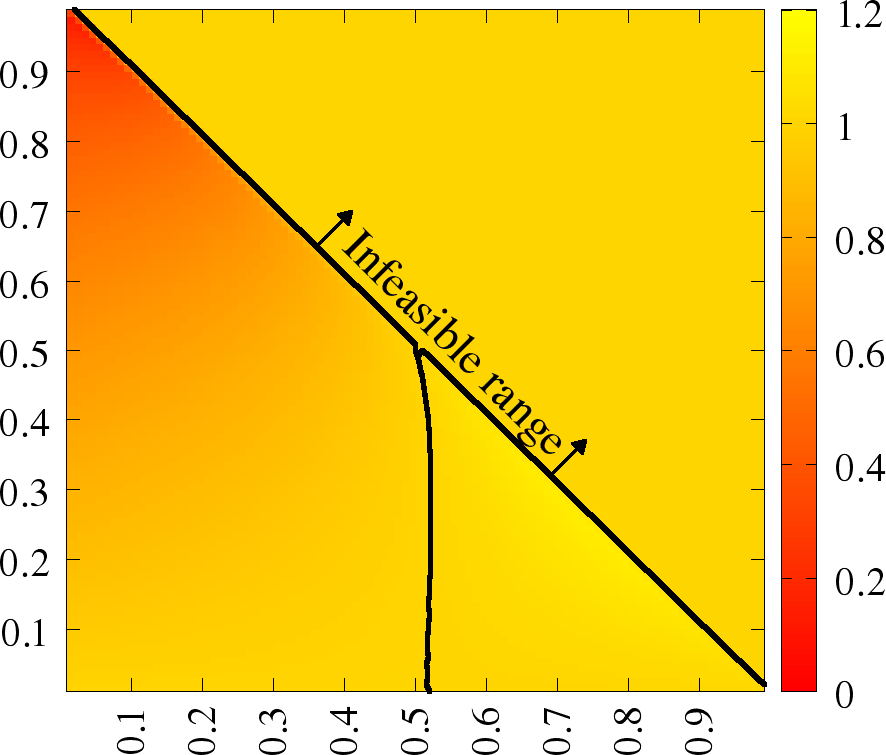}
\label{fig:twoPools:r1}
}
\hfil
\subfloat[$r_2$]{
\includegraphics[width=0.4\linewidth]{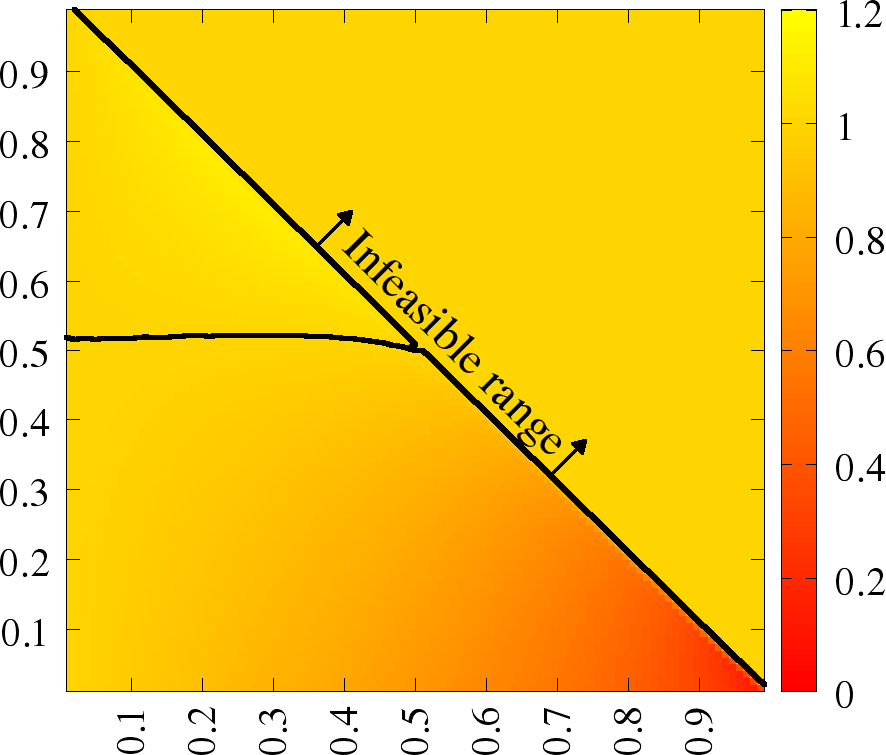}
\label{fig:twoPools:r2}
}

\caption[.]{\protect
Two attacking pools system: Optimal infiltration rates ($x_1$ and $x_2$) and corresponding revenues ($r_1$ and $r_2$) as a function of pool sizes. Lines in~\subref{fig:twoPools:x1} and~\subref{fig:twoPools:x2} are at $\xij{1}{2}=0$ and $\xij{2}{1}=0$, respectively. Lines in~\subref{fig:twoPools:r1} and~\subref{fig:twoPools:r2} are at $r_1=1$ and $r_2=1$, respectively. 
}
\label{fig:twoPools}
\end{figure*}

We proceed to analyze the case where two pools may attack each other and the other miners mine solo. Again we have pool~1 of size~$m_1$ and pool~2 of size~$m_2$; pool~1 controls its infiltration rate $\xij{1}{2}$ of pool~2, but now pool~2 also controls its infiltration rate~$\xij{2}{1}$ of pool~1. 
This scenario is illustrated in Figure~\ref{fig:twoPoolsIllustration} 

\begin{figure}[!t]
\centering
\includegraphics[width=\linewidth]{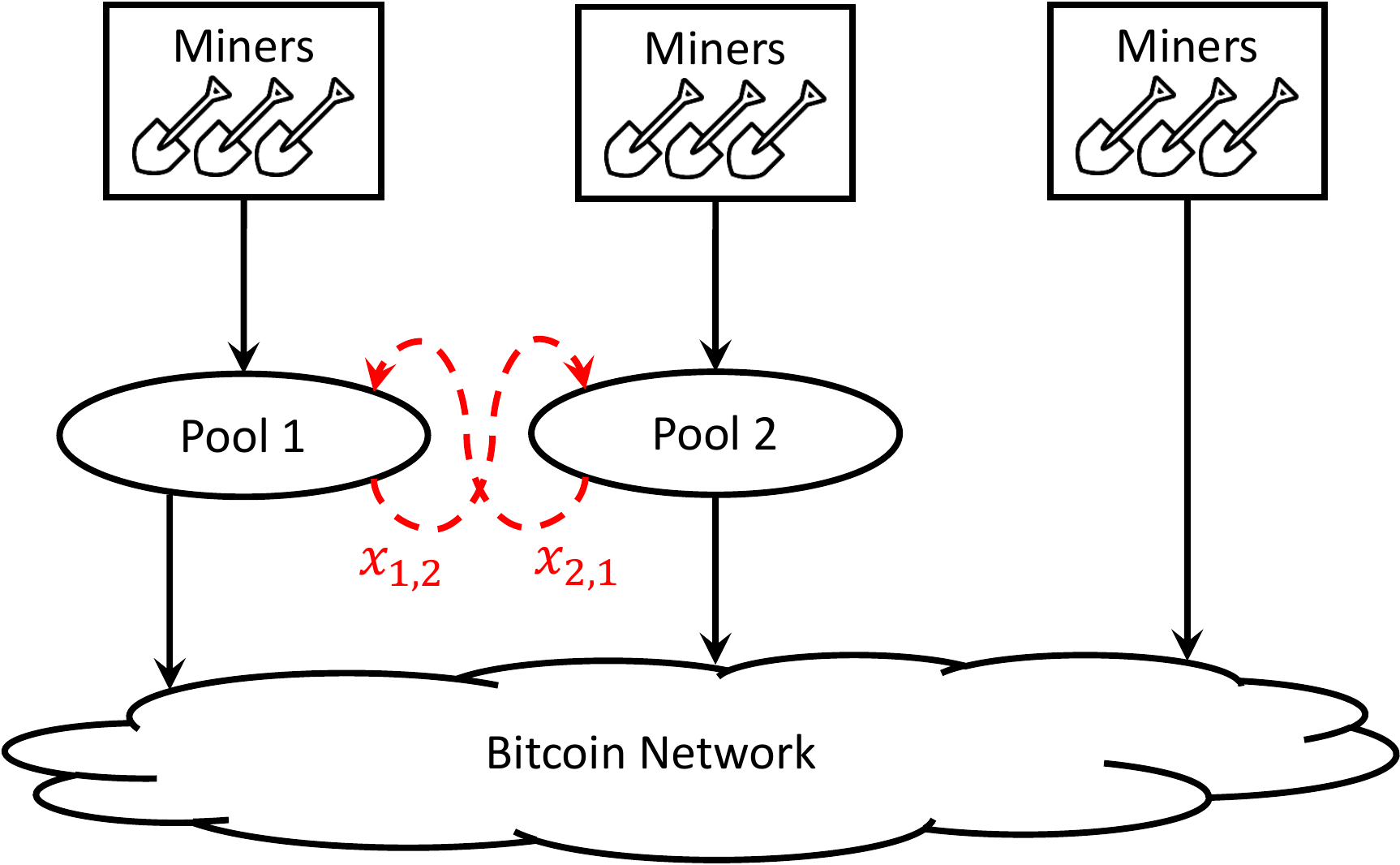}
\caption[.]{\protect 
Two pools attacking each other. 
} 
\label{fig:twoPoolsIllustration}
\end{figure}

The total mining power in the system is $m - \xij{1}{2} - \xij{2}{1}$
The direct revenues $R_1$ and $R_2$ of the pools from mining are their effective mining rates, without infiltrating mining power, divided by the total mining rate. 
\begin{equation} 
\begin{aligned} 
R_1 = \frac{m_1 - \xij{1}{2}}{m - \xij{1}{2} - \xij{2}{1}} &&\\ 
R_2 = \frac{m_2 - \xij{2}{1}}{m - \xij{1}{2} - \xij{2}{1}} && \,\,\, .  
\end{aligned} 
\end{equation} 

The total revenue of each pool is its direct mining revenue, above, and the infiltration revenue from the previous round, which is the attacked pool's total revenue multiplied by its infiltration rate. 
The pool's total revenue is divided among its loyal miners and miners that infiltrated it. At stable state this is 
\begin{equation} 
\begin{aligned} 
r_1 = \frac{R_1 + \xij{1}{2} r_2}{m_1 + \xij{2}{1}} && \\
r_2 = \frac{R_2 + \xij{2}{1} r_1}{m_2 + \xij{1}{2}} && \,\,\, .
\end{aligned}
\end{equation} 

Solving for $r_1$ and $r_2$ we obtain the following closed expressions for each. We express the revenues as functions of $\xij{1}{2}$ and $\xij{2}{1}$. 
\begin{equation} \label{eqn:twoAttackers:rs} 
\begin{aligned} 
r_1(\xij{1}{2}, \xij{2}{1}) 
= 
\frac{m_2 R_1+\xij{1}{2} (R_1+R_2)
}{
m_1 m_2 + m_1 \xij{1}{2} + m_2 \xij{2}{1}
} && 
\\
r_2(\xij{2}{1}, \xij{1}{2}) 
= 
\frac{
m_1 R_2+\xij{2}{1} (R_1+R_2)
}{
m_1 m_2 + m_1 \xij{1}{2} + m_2 \xij{2}{1}
} && 
\,\,\, .  
\end{aligned} 
\end{equation}

Each pool controls only its own infiltration rate. In each round of the pool game, each pool will optimize its infiltration rate of the other. If pool~1 acts at step~$t$, it optimizes its revenue with  
\begin{equation} \label{eqn:twoAttackers:x12}
\xij{1}{2}(t) \gets \arg\max_{x'} r_1(x', \xij{2}{1}(t-1)) \,\,\, , 
\end{equation} 
and if pool~2 acts at step~$t$, it optimizes its revenue with 
\begin{equation} \label{eqn:twoAttackers:x21}
\xij{2}{1}(t) \gets \arg\max_{x'} r_2(x', \xij{1}{2}(t-1)) \,\,\, .  
\end{equation} 

An equilibrium exists where neither pool~1 nor pool~2 can improve its revenue by changing its infiltration rate. That is, any pair of values $x_1', x_2'$ such that 
\begin{equation} \label{eqn:twoPoolsProblem}
\left\{
\begin{array}{l}
\arg\max_{\xij{1}{2}} r_1(\xij{1}{2}, \xij{2}{1}') = \xij{1}{2}' \\ 
\arg\max_{\xij{2}{1}} r_2(\xij{1}{2}', \xij{2}{1}) = \xij{2}{1}' 
\end{array}
\right.
\end{equation} 
under the constraints 
\begin{equation} \label{eqn:twoPoolsConstraints} 
\begin{aligned} 
0 < x_1 < m_1 && \\
0 < x_2 < m_2 && \,\,\, . 
\end{aligned}
\end{equation} 

The feasible region for the pool sizes is $m_1 > 0, m_2 > 0$, and $m_1 + m_2 \le m$. 
The revenue function for $r_i$ is concave in $x_i$ ($\partial^2 r_i / \partial x_i^2 < 0$) for all feasible values of the variables. Therefore the solutions for equations~\ref{eqn:twoAttackers:x12} and~\ref{eqn:twoAttackers:x21} are unique and are either at the borders of the feasible region or where $\partial r_i / \partial \xij{i}{j} = 0$. 

From Section~\ref{sec:twoPoolsOneAttacker} we know that no-attack is not an equilibrium point, since each pool can increase its revenue by choosing a strictly positive infiltration rate, that is, $\xij{1}{2} = \xij{2}{1} = 0$ is not a solution to Equations~\ref{eqn:twoPoolsProblem}--\ref{eqn:twoPoolsConstraints}. 

Nash equilibrium therefore exists with \xij{1}{2}, \xij{2}{1} values where 
\begin{equation} \label{eqn:twoPoolsNE}
\left\{
\begin{array}{l} 
\displaystyle
\frac{\partial r_1(\xij{1}{2}, \xij{2}{1})}{\partial{\xij{1}{2}}} = 0 \\ 
\displaystyle
\frac{\partial r_2(\xij{2}{1}, \xij{1}{2})}{\partial{\xij{2}{1}}} = 0 
\end{array} 
\right.
\,\,\, . 
\end{equation}

Using symbolic computation tools, we see that there is a single pair of values for which Equation~\ref{eqn:twoPoolsNE} holds for a choice of $m_1$ and $m_2$. 

        \subsection{Numerical Analysis} 

A numerical analysis confirms these observations. We simulate the pool game for a range of pool sizes. 
For each choice of pool sizes, we start the simulation when both pools do not infiltrate each other, $\xij{1}{2} = \xij{2}{1} = 0$, and the revenue densities are $r_1 = r_2 = 1$. 
At each round one pool chooses its optimal infiltration rate based on the pool sizes and the rate with which it is infiltrated, and we calculate the revenue after convergence with Equation~\ref{eqn:twoAttackers:rs}. Recall the players in the pool game are chosen with the Round Robin policy, so the pools take turns, and we let the game run until convergence. The results are illustrated in Figure~\ref{fig:twoPools}. 

Each run with a some $m_1, m_2$ values results in a single point in each graph in Figure~\ref{fig:twoPools}. We depict the infiltration rates of both pools $\xij{1}{2}, \xij{2}{1}$ in Figures~\ref{fig:twoPools:x1}--\ref{fig:twoPools:x2} and the pools' revenue densities~$r_1, r_2$ in Figures~\ref{fig:twoPools:r1}--\ref{fig:twoPools:r2}. 
So for each choice of $m_1$ and $m_2$, the values of $\xij{1}{2}$, $\xij{2}{1}$, $m_1$ and $m_2$ are the points in each of the graphs with the respective coordinates. 

As before, for the $\xij{i}{j}$ graphs we draw a border around the region where there is no-attack by $i$ in equilibrium. For the $r_i$ graphs we draw a line around the region where the revenue is the same as in the no-attack scenario, namely~1. 

We first observe that only in extreme cases a pool does not attack its counterpart. 
Specifically, at equilibrium a pool will refrain from attacking only if the other pool is larger than about $80\%$ of the total mining power. 

But, more importantly, we observe that a pool improves its revenue compared to the no-pool-attacks scenario only when it controls a strict majority of the total mining power. 
These are the small triangular regions in Figures~\ref{fig:twoPools:r1} and~\ref{fig:twoPools:r2}. 

        \subsection{The Prisoner's Dilemma} 

% \begin{SCtable*}
\begin{figure*}[t]
\begin{tabular}{|l|c|c|}
\hline
\backslashbox{Pool 2\kern-0.5em}{Pool 1} & no attack & attack \\
\hline
no attack & $(r_1 = 1, r_2 = 1)$ & $(r_1 > 1, r_2 = \tilde{r}_2 < 1)$ \\
\hline
attack & $(r_1 = \tilde{r}_1 < 1, r_2 > 1)$ & $(\tilde{r}_1 < r_1 <1 , \tilde{r}_2 < r_2 < 1)$ \\
\hline
\end{tabular}
\caption{Prisoner's Dilemma for two pools. The revenue density of each pool is determined by the decision of both pools whether to attack or not. 
The dominant strategy of each player is to attack, however the payoff of both would be larger if they both refrain from attacking.} 
\label{tbl:prisoners}
% \end{SCtable*} 
\end{figure*} 

In a stable Bitcoin environment with two pools, where neither controls a strict majority of the mining power, both pools will earn less at equilibrium than if both pools ran without attacking. 
We can analyze in this case a game where each pool chooses either to attack, or not to attack. If it attacks, the pool optimizes its revenue. 

Consider pool~1 without loss of generality. 
As we have seen in Section~\ref{sec:twoPoolsOneAttacker}, if pool~2 does not attack, pool~1 can increase its revenue above~1 by attacking, setting~$\xij{1}{2}$ to the optimum. 
If pool~2 does attack but pool~1 does not, we denote the revenue of pool~1 by $\tilde{r}_1$. The exact value of $\tilde{r}_1$ depends on the values of~$m_1$ and~$m_2$, but it is always smaller than one. As we have seen above, if pool~1 does choose to attack, its revenue increases, but does not surpass one. The game is summarized in Figure~\ref{tbl:prisoners}. 

When played once, this is the classical prisoner's dilemma. Attack is the dominant strategy: Whether pool~2 chooses to attack or not, the revenue of pool~1 is larger when attacking than when refraining from attack, and the same for pool~2. 
At equilibrium of this attack-or-don't game, when both pools attack, the revenue of each pool is smaller than its revenue if neither pool attacked. 

However, the game is not played once, but rather continuously, forming a super-game, where each pool can change its strategy between attack and no-attack. 
If the pools agree (even implicitly) to coordinate, in each round a pool can detect whether it is being attacked and deduce that the other pool is violating the agreement. 
In this super-game, cooperation where neither pool attacks is a possible stable state~\cite{friedman1971nonCooperative,aumann1994longTerm} despite the fact that the single Nash equilibrium in every round is to attack. 

%%%%%%%%%%%%%%%%%%%%%%%%%%%%%%%%%%%%%%%%%%%%%%%%%%%%%%%%%%%%%%%%%%%%%%%%%%%%%%% 
%%%%%%%%%%%%%%%%%%%%%%%%%%%%%%%%%%%%%%%%%%%%%%%%%%%%%%%%%%%%%%%%%%%%%%%%%%%%%%% 
%%%%%%%%%%%%%%%%%%%%%%%%%%%%%%%%%%%%%%%%%%%%%%%%%%%%%%%%%%%%%%%%%%%%%%%%%%%%%%% 

    \section{\texorpdfstring{$p$}{p} Identical Pools} \label{sec:pPools}

Let there be any number of pools of identical size that engage in block withholding against one another. 
In this case there exists a symmetric equilibrium. 
% In a symmetric equilibrium pools attack one another with identical rates. 
Consider, without loss of generality, a step of pool~1. 
It controls its attack rates each of the other pools, and due to symmetry they are all the same. 
Denote by $\xij{1}{\lnot 1}$ the attack rate of pool~1 against any other pool. 
Each of the other pools cab attack its peers as well. 
Due to symmetry, all attack rates by all attackers are identical. 
Denote by $\xij{\lnot 1}{*}$ the attack rate of any pool other than~1 against any other pool, including pool~1. 

Denote by $R_1$ the direct revenue (from mining) of pool~1 and by $R_{\lnot 1}$ the direct revenue of each of the other pools. Similarly denote by $r_1$ and $r_{\lnot 1}$ the revenue densities of pool~1 and other pools, respectively. 

The generic equations~\ref{eqn:RiFull} and~\ref{eqn:riFull} are instantiated to 
\begin{equation} \label{eqn:RsSymm} 
\begin{aligned} 
&R_1 
= 
\frac{
    m_i - (p-1) \xij{1}{\lnot 1}
}{
    m - (p - 1) (p - 1) \xij{\lnot 1}{*} - (p - 1) \xij{1}{\lnot 1}
} 
\\
&R_{\lnot 1} 
= 
\frac{ 
    m_i - (p - 1) \xij{\lnot 1}{*} 
}{ 
    m - (p - 1) (p - 1) \xij{\lnot 1}{*} - (p - 1) \xij{1}{\lnot 1} 
} 
\end{aligned}
\end{equation}
and  
\begin{equation} \label{eqn:rsSymm}
\begin{aligned} 
&r_1 
= 
\frac{
R_1 + (p - 1) \xij{1}{\lnot 1} r_{\lnot 1} 
}{ 
m_i + (p - 1) \xij{\lnot 1}{1} 
} 
\\
&r_{\lnot 1}  
= 
\frac{
R_{\lnot 1} + (p - 2) \xij{\lnot 1}{*} r_{\lnot 1} + \xij{\lnot 1}{*} r_1
}{ 
m_i + (p - 2) \xij{\lnot 1}{*} + \xij{1}{\lnot 1} 
} 
\end{aligned}
\,\,\, . 
\end{equation} 

Substituting Equations~\ref{eqn:RsSymm} in Equation~\ref{eqn:rsSymm} and solving we obtain a single expression for any $r_i$, since in the symmetric case we have $r_1 = r_{\lnot 1}$. The expression is shown in Equation~\ref{eqn:rsSymm} (Figure~\ref{fig:r1Symm}). 

\begin{figure*} 
\begin{equation} \label{eqn:r1Symm} 
r_i
=
-\frac{m_i^2+m_i \xij{1}{\lnot 1}-(p-1) \xij{1}{\lnot 1} ((p-1) \xij{\lnot 1}{*}+\xij{1}{\lnot 1})}{\left((p-1) \xij{1}{\lnot 1}+(p-1)^2 \xij{\lnot 1}{*}-1\right) ((m_i+\xij{1}{\lnot 1}) (m_i+(p-1) \xij{\lnot 1}{1})-(p-1) \xij{1}{\lnot 1} \xij{\lnot 1}{*})} 
\end{equation}
\caption{ 
Expression for $r_i$ in a system with pools of equal size. 
} 
\label{fig:r1Symm}
\end{figure*}

Given any value of $p$ and $m_i$ (where $p m_i < 1$), the feasible range of the infiltration rates is $0 \le \xij{i}{j} \le m_i / p$. Within this range $r_i$ is continuous, differentiable, and concave in $\xij{1}{\lnot 1}$. 
Therefore, the optimal point for pool~1 is where $\partial r_1 / \partial \xij{1}{\lnot 1} = \nobreak 0$. 
Since the function is concave the equation yields a single feasible solution, which is a function of the attack rates of the other pools, namely $\xij{\lnot 1}{1}$ and $\xij{\lnot 1}{*}$. 

To find a symmetric equilibrium, we equate $\xij{1}{\lnot 1} = \xij{\lnot 1}{1} = \xij{\lnot 1}{*}$ and obtain a single feasible solution. The equilibrium infiltration rate and the matching revenues are shown in Equation~\ref{eqn:ppoolsStableX} (Figure~\ref{fig:ppoolsStableX}). 

\begin{figure*} 
\begin{equation} \label{eqn:ppoolsStableX}
\begin{aligned} 
&\xijbar{1}{\lnot 1} 
= 
\xijbar{\lnot 1}{1} 
= 
\xijbar{\lnot 1}{*}
=
\frac{
p - m_i - \sqrt{(m_i - p)^2 - 4 (m_i)^2 (p - 1)^2 p}
}{
2 (p - 1)^2 p)
}
\\ 
&\bar{r}_1 
= 
\bar{r}_{\lnot 1} 
=
\frac{
    2 p
}{
    p - m_i + 2 m_i p + \sqrt{(m_i - p)^2 - 4 (m_i)^2 (p - 1)^2 p}
}
\end{aligned}
\end{equation}
\caption{
Symmetric equilibrium values for a system of $p$~pools of equal sizes. 
}
\label{fig:ppoolsStableX} 
\end{figure*} 

As in the two-pool scenario, the revenue at the symmetric equilibrium is inferior to the no-one-attacks non-equilibrium strategy.

    \section{Discussion} \label{sec:discussion} 

        \subsection{Bitcoin's Health} \label{sec:implications} 
    
Large pools hinder Bitcoin's distributed nature as they put a lot of mining power in the hands of a few pool managers. This has been mostly addressed by community pressure on miners to avoid forming large pools~\cite{andresen2014centralized}. However such recommendations had only had limited success, and mining is still dominated by a small number of large pools. 
As a characteristic example, in the period of November 2--8, 2014, three pools generated over $50\%$ of the proofs of work~\cite{organofcorti2014poolStats}. 

Long term block withholding attacks are difficult to hide, since miners using an attacked pool would notice the reduced revenue density. 
Nevertheless, such attacks are rarely reported\footnote{A recent example is an attack that was partially subverted due to limited efforts of the attacker to hide itself and an alert pool manager~\cite{wizkid2013eligius}. It is unknown whether this was a classical block withholding attack or a more elaborate scheme.}, and we can therefore conclude that they are indeed rare. 
The fact that such attacks do not persist may indicate that the active pools have reached an implicit or explicit agreement not to attack one another. 

However, an attacked pool cannot detect which of its miners are attacking it, let alone which pool controls the miners. 
At some point a pool might miscalculate and decide to try and increase its revenue. 
One pool might be enough to break the agreement, possibly leading to a constant rate of attacks among pools and a reduced revenue. 

If open pools reach a state where their revenue density is reduced due to attacks, miners will leave them in favor of other available options: Miners of sufficient size can mine solo; smaller miners can form private pools with closed access, limited to trusted participants. 

Such a change in the mining forces may be in favor of Bitcoin as a whole. Since they require such intimate trust, we believe private pools are likely to be smaller, and lead to a fine grained distribution of mining power with many small pools and solo miners. 

        \subsection{Miners and Pools} 

            \subsubsection{Direct Pool Competition} 

Since miners evidently prefer to work with public pools rather than solo~\cite{organofcorti2014poolStats}, a pool may engage in an attack against another pool not to increase its absolute revenue, but rather to attract miners by temporarily increasing its revenue compared to a competing pool. 

Our analysis addressed the eventual revenue of pools under block withholding attacks, after the Bitcoin system has normalized the revenues by adjusting difficulty. 
Before this normalization, the revenue of an attacking pool is reduced due to the reduction in revenue of both the attacking and attacked pools. 
Nevertheless, the attacker's revenue density compared to the victim's revenue density is immediately improved. 

This is an an enhanced version of the classical sabotage block withholding with a lower overhead for the attacker. Eventually, once difficulty is adjusted, the attacker may see an absolute benefit in attacking. 

The pool game model does not cover the dynamic interplay of pools and miners, which we leave for future work. 

            \subsubsection{Pool Fees} 

We assumed in our analysis that pools do not charge fees from their members since such fees are typically nominal ($0$ -- $3\%$ of a pool's revenue~\cite{btcWiki2014pools}). 
The model can be extended to include pools fees. 
Fees would add a friction element to the flow of revenue among infiltrated and infiltrating pools. 
Specifically, Equation~\ref{eqn:riFull} would change to take into account a pool fee of $f$ 
\begin{equation} \label{eqn:riFullWithFee}
r_i(t) 
= 
\frac{
R_i(t) + \sum_{j = 1}^{p} \xij{i}{j}(t) (1 - f) r_j(t)
}{
m_i + \sum_{j = 1}^{p} \xij{j}{i}(t) 
}
\,\,\, . 
\end{equation} 

A pool with a fee of $f$ is a less attractive target for block withholding, since the attacker's revenue is reduced by $f$.
However it is also less attractive for miners in general. 
Trading off the two for best protection is left for future work, as part of the treatment of the miner-pool interplay. 

%%%%%%%%%%%%%%%%%%%%%%%%%%%%%%%%%%%%%%%%%%%%%%%%%%%%%%%%%%%%%%%%%%%%%%%%%%%%%%%
%%%%%%%%%%%%%%%%%%%%%%%%%%%%%%%%%%%%%%%%%%%%%%%%%%%%%%%%%%%%%%%%%%%%%%%%%%%%%%%
%%%%%%%%%%%%%%%%%%%%%%%%%%%%%%%%%%%%%%%%%%%%%%%%%%%%%%%%%%%%%%%%%%%%%%%%%%%%%%%

    \section{Related Work} \label{sec:related} 

        \subsection{The Block Withholding Attack} 

The danger of a block withholding attack is as old as Bitcoin pools. 
The attack was described by Rosenfeld~\cite{rosenfeld2011analysis} as early as~2011, as pools were becoming a dominant player in the Bitcoin world. 
The paper described the standard attack, used by a miner to sabotage a pool at the cost of reducing its own revenue. Early work did not address the possibility of pools infiltrating other pools for block withholding. 

Courtois and Bahack~\cite{courtois2014subversive} have recently noted that a pool can increase its overall revenue with block withholding if all other mining is performed by honest pools. 
We consider the general case where not all mining is performed through public pools, and analyze situations where pools can attack one another. 
The discrepancy between the calculations of~\cite{courtois2014subversive} and our results for the special case analyzed there can be explained by the strong approximations in that work. 
For example, we calculate exactly how infiltrating miners reduce the revenue density of the infiltrated pool. 

        \subsection{Temporary Block Withholding} 

In the Block withholding attack discussed in this work the withheld blocks are never published. 
However, blocks can be withheld temporarily, not following the Bitcoin protocol, to improve an attacker's revenue. 

An attacker can perform a double spending attack as follows~\cite{rosenfeld2011analysis}. He intentionally generates two conflicting transactions, places one in a block it withholds, and publishes the other transaction. 
After the recipient sees the published transaction, the attacker publishes the withheld block to revoke the former transaction. This attack is performed by miners or pools against service providers that accept Bitcoin, and it unrelated to this work. 

A miner or a pool can perform a selfish mining attack. Here, the attacker increases its revenue by temporarily withholding its blocks and publishing them in response to block publication by other pools and miners~\cite{eyal2013majority}. This attack is independent of the block withholding attack we discuss here and the two can be performed concurrently. 

        \subsection{Block Withholding Defense} 

Most crypto-currencies use a proof-of-work architecture similar to Bitcoin, where finding proof of work is the result of solution guessing and checking. All of the algorithms we are aware of are susceptible to the block withholding attack, as in all of them the miner can check whether she found a full solution or a partial proof of work. 
Prominent examples are Litecoin~\cite{litecoin2013site}, Dogecoin~\cite{dogecoin2013site} and Permacoin~\cite{miller2014permacoin}. 

Rosenfeld~\cite{rosenfeld2011analysis} suggested a change of the block structure that would allow a pool to probe for block withholding with a honey-pot technique. A pool could generate miner tasks that it knows would lead to a (useless) block solution. An attacker would withhold the solution and expose itself. 

This fix, a different proof of work algorithm, or another solution, could reduce or remove the danger of block withholding. However, this may not be in the interest of the community: Pool block withholding, or even its potential, could lead to a reduction of pool sizes, as explained in Section~\ref{sec:implications}. 

        \subsection{Decentralized Pools} 

Although most pools use a centralized manager, a prominent exception is P2Pool~-- a distributed pool architecture with no central manager~\cite{forrsetv2011p2pool}. 
But the question of whether a pool is run by a centralized manager or with a decentralized architecture is almost immaterial for the attack we describe. 
An open P2Pool group can be infiltrated and attacked, and the P2Pool code can be changed to support attacks against other pools. 

On the other hand, P2Pool can be used by groups of miners to easily form closed pools. These do not accept untrusted miners, and are therefore protected against block withholding. 

%%%%%%%%%%%%%%%%%%%%%%%%%%%%%%%%%%%%%%%%%%%%%%%%%%%%%%%%%%%%%%%%%%%%%%%%%%%%%%%
%%%%%%%%%%%%%%%%%%%%%%%%%%%%%%%%%%%%%%%%%%%%%%%%%%%%%%%%%%%%%%%%%%%%%%%%%%%%%%%
%%%%%%%%%%%%%%%%%%%%%%%%%%%%%%%%%%%%%%%%%%%%%%%%%%%%%%%%%%%%%%%%%%%%%%%%%%%%%%%

    \section{Conclusion} \label{sec:conclusion}

We explored a block withholding attack among Bitcoin mining pools~--- an attack that is possible in any similar system that rewards for proof of work. Such systems are gaining popularity, running most digital currencies and related services. 

We observe that no-pool-attacks is not a Nash equilibrium: If none of the other pools attack, a pool can increase its revenue by attacking the others. 

When two pools can attack each other, they face a version of the Prisoner's Dilemma. If one pool chooses to attack, the victim's revenue is reduced, and it can retaliate by attacking and increase its revenue. However, when both attack, at Nash equilibrium both earn less than they would have if neither attacked. With multiple pools of equal size a similar situation arises with a symmetric equilibrium. 

The fact that block withholding is not common may be explained by modeling the attack decisions as an iterative prisoner's dilemma. However, we argue that since the attack can be done anonymously by any of the pools, this situation is unstable. Eventually one pool may decide to increase its revenue and drag the others to attack as well, ending with a reduced revenue for all. This would push miners to join private pools which can verify that their registered miners do not withhold blocks. This may lead to smaller pools, and so ultimately to a better environment for Bitcoin as a whole. 

\vspace{1ex}
\paragraph*{Acknowledgements} 
The author is grateful to 
Ken Birman, 
Fred B.\ Schneider, and
Eva Tardos for their valuable advice. 

%%%%%%%%%%%%%%%%%%%%%%%%%%%%%%%%%%%%%%%%%%%%%%%%%%%%%%%%%%%%%%%%%%%%%%%%%%%%%%%
%%%%%%%%%%%%%%%%%%%%%%%%%%%%%%%%%%%%%%%%%%%%%%%%%%%%%%%%%%%%%%%%%%%%%%%%%%%%%%%
%%%%%%%%%%%%%%%%%%%%%%%%%%%%%%%%%%%%%%%%%%%%%%%%%%%%%%%%%%%%%%%%%%%%%%%%%%%%%%%

\vfill\eject 

\bibliographystyle{plain} 
\bibliography{btc} 

%%%%%%%%%%%%%%%%%%%%%%%%%%%%%%%%%%%%%%%%%%%%%%%%%%%%%%%%%%%%%%%%%%%%%%%%%%%%%%%
%%%%%%%%%%%%%%%%%%%%%%%%%%%%%%%%%%%%%%%%%%%%%%%%%%%%%%%%%%%%%%%%%%%%%%%%%%%%%%%
%%%%%%%%%%%%%%%%%%%%%%%%%%%%%%%%%%%%%%%%%%%%%%%%%%%%%%%%%%%%%%%%%%%%%%%%%%%%%%%

\end{document}